\documentclass[a4paper,UKenglish]{article}

\usepackage[margin=1in]{geometry}
\usepackage{microtype}
\usepackage[utf8]{inputenc}
\usepackage{amsthm}
\usepackage{amsmath}
\usepackage{amsfonts}
\usepackage{amssymb}
\usepackage{graphicx}
\usepackage{subcaption}
\usepackage{hyperref}
\usepackage[noabbrev,capitalise]{cleveref}
\usepackage{xspace}
\usepackage{cite}
\usepackage{bm}
\usepackage{thm-restate}
\usepackage{mathtools}
\usepackage[noabbrev,capitalise]{cleveref}

\newcommand{\N}{\mathbb{N}}

\newcommand{\R}{\mathbb{R}}
\newcommand{\Z}{\mathbb{Z}}

\newcommand{\V}{\mathbb{V}}

\newcommand{\av}{$\alpha$-vector\xspace}
\newcommand{\dv}{dist-value\xspace}

\newcommand{\dotp}[2]{\langle#1,#2\rangle}

\newcommand{\ignore}[1]{}

\newcommand{\conv}{\mathrm{conv}}
\newcommand{\poly}{\mathrm{poly}}

\newcommand{\Caratheodory}{Carath\'eodory\xspace}
\newcommand{\Barany}{B\'ar\'any\xspace}

\newcommand{\Jeronimo}{Jer\'onimo\xspace}
\newcommand{\Matousek}{Matou\v{s}ek\xspace}

\DeclarePairedDelimiter{\ceil}{\lceil}{\rceil}
\DeclarePairedDelimiter{\floor}{\lfloor}{\rfloor}

\def\cc#1{\mathsf{#1}}
\def\cUEOPL{\ensuremath{\cc{UEOPL}}\xspace}
\def\TFNP{\ensuremath{\cc{TFNP}}\xspace}
\def\PPA{\ensuremath{\cc{PPA}}\xspace}
\def\PPAD{\ensuremath{\cc{PPAD}}\xspace}
\def\PLS{\ensuremath{\cc{PLS}}\xspace}
\def\CLS{\ensuremath{\cc{CLS}}\xspace}

\def\NP{\ensuremath{\cc{NP}}\xspace}
\def\coNP{\ensuremath{\cc{coNP}}\xspace}

\def\problem#1{{\scshape #1}}
\def\UEOPL{\problem{UniqueEOPL}\xspace}
\def\GHS{\problem{Alpha-HS}\xspace}
\def\HS{\problem{Ham-Sandwich}\xspace}

\renewcommand{\S}{\mathbb{S}}
\renewcommand{\P}{\mathbb{P}}
\renewcommand{\V}{\mathbb{V}}

\crefname{Step}{Step}{Steps}

\newtheorem{theorem}{Theorem}
\newtheorem{definition}[theorem]{Definition}
\newtheorem{lemma}[theorem]{Lemma}
\theoremstyle{plain}

\newtheorem{claim}[theorem]{Claim}

\newtheorem{remark}[theorem]{Remark}

\title{Computational Complexity of the \texorpdfstring{$\alpha$}{a}-Ham-Sandwich Problem\thanks{Supported in
		part by ERC StG 757609.}}
\author{
	Man-Kwun Chiu\thanks{Institut f\"ur Informatik, 
		Freie Universit\"at Berlin. \texttt{[chiumk,arunich,mulzer]@inf.fu-berlin.de}.
	}
	\and 
	Aruni Choudhary\footnotemark[2]
	\and
	Wolfgang Mulzer\footnotemark[2]
}

\begin{document}
	
\maketitle
\begin{abstract}
The classic Ham-Sandwich theorem states
that for any $d$ measurable sets in $\R^d$, there 
is a hyperplane that bisects them simultaneously. 
An extension by \Barany, Hubard,
and \Jeronimo\,[DCG\,2008] states that if 
the sets are convex and \emph{well-separated},
then for any given 
$\alpha_1, \dots, \alpha_d \in [0, 1]$,
there is a unique oriented hyperplane that cuts off 
a respective fraction $\alpha_1, \dots, \alpha_d$ 
from each set. Steiger and Zhao\,[DCG\,2010] proved 
a discrete analogue of this theorem, which we call
the \emph{$\alpha$-Ham-Sandwich theorem}.
They gave an algorithm to find the hyperplane
in time $O(n (\log n)^{d-3})$, 
where $n$ is the total number of input points.
The computational complexity of this search problem 
in high dimensions is open, quite unlike 
the complexity of the Ham-Sandwich problem, which 
is now known
to be \PPA-complete (Filos-Ratsikas and Goldberg [STOC 2019]).

Recently, Fearley, Gordon, Mehta, and Savani\,[ICALP\,2019]
introduced a new sub-class of \CLS (Continuous Local Search) called
\emph{Unique End-of-Potential Line} (\cUEOPL). This class captures problems
in \CLS that have unique solutions.
We show that for the $\alpha$-Ham-Sandwich theorem, 
the search problem of finding the dividing hyperplane lies in \cUEOPL.
This gives the first non-trivial containment of the 
problem in a complexity class and places it in the 
company of classic search problems
such as finding the fixed point of a contraction map, 
the unique sink orientation problem
and the $P$-matrix linear complementarity problem.
	
\end{abstract}

\section{Introduction}
\label{section:introduction}

\subparagraph*{Motivation and related work.}
The Ham-Sandwich Theorem~\cite{st-sandwich} 
is a classic result about partitioning 
sets in high dimensions: for any $d$ measurable sets 
$S_1, \dots, S_d \subset \R^d$ in $d$ dimensions,
there is an oriented hyperplane $H$ that  
simultaneously \emph{bisects} $S_1,\dots,S_d$.
More precisely, if $H^{+},H^{-}$ are the
closed half-spaces bounded by $H$, then for $i = 1, \dots, d$,
the measure of 
$S_i \cap H^+$ equals the measure of $S_i \cap H^-$. 
The traditional proof goes through the Borsuk-Ulam 
Theorem~\cite{mt-bubook}.
The Ham-Sandwich Theorem is a
cornerstone of geometry and topology, 
and it has found applications in other
areas of mathematics,
e.g., for the study of majority rule voting and the analysis 
of the stability of bicameral legislatures in 
social choice theory~\cite{cm-social}.

Let $[n] = \{1,\dots,n\}$. The \emph{discrete}
Ham-Sandwich 
Theorem~\cite{mt-bubook,lms-hs} states that for any 
$d$ finite point sets 
$P_1,\dots,P_d \subset \R^d$ in $d$ dimensions, 
there is an oriented hyperplane $H$ such that 
$H$ bisects each $P_i$, i.e., for 
$i \in[d]$, we have 
$\min \{|P_i \cap H^+|, |P_i \cap H^-|\} \ge  \left\lceil |P_i|/2 \right\rceil $. 
We denote the associated search problem as \HS.
Lo, \Matousek, and 
Steiger~\cite{lms-hs} gave an $n^{O(d)}$-time algorithm for \HS.
They also provided a linear-time algorithm 
for points in $\R^3$, under additional constraints.

There are many alternative and more 
general variants of both the continuous and
the discrete Ham-Sandwich Theorem.
For example, \Barany and \Matousek~\cite{bm-fans} derived 
a version where measures in the plane can be divided 
into any (possibly different) ratios by \emph{fans} instead of 
hyperplanes (lines). A discrete variant of this 
result was given by Bereg~\cite{bereg-fans}.
Schnider~\cite{schnider-wedge} studied a generalization 
in higher dimensions.  Recently Barba, Pilz, and 
Schnider~\cite{bps-pizza} showed that four measures 
in the plane can be bisected with two lines.
Zivaljevi\'{c} and Vre\'{c}ica~\cite{zv-extension} 
proved a result that interpolates between the Ham-Sandwich Theorem
and the Centerpoint Theorem~\cite{rado-centerpoint}, of which
there is also a no-dimensional version~\cite{cm-tverberg}.
Schnider~\cite{schnider-hsct} presented a generalization 
based on this result among others.

Here, we focus on a version that allows for
dividing the sets into arbitrary given ratios instead 
of simply bisecting them. The sets $S_1,\dots,S_d \subset \R^d$ 
are \emph{well-separated} if every selection 
of them can be strictly separated from the others by a 
hyperplane. \Barany, Hubard, and \Jeronimo~\cite{bhj-ghs}
showed that if $S_1, \dots, S_d$ are well-separated and convex, 
then for any given reals $\alpha_1,\dots,\alpha_d \in [0,1]$, 
there is a unique hyperplane that divides 
$S_1,\dots,S_d$ in the ratios $\alpha_1,\dots,\alpha_d$, 
respectively. Their proof goes through Brouwer's Fixed 
Point Theorem. Steiger and Zhao~\cite{sz-ghs} 
formulated a discrete version.
In this setup, $S_1, \dots, S_d$ are finite point sets.
Again, we need that the (convex hulls of
the) $S_i$ are well-separated.
Additionally, we require that the $S_i$ 
follow a weak version of general position.
Let $\alpha_1, \dots, \alpha_d \in \N$ be $d$ integers with
$1 \leq \alpha_i\leq |S_i|$, for $i \in [d]$.
Then, there is a unique oriented hyperplane $H$
that passes through one point from each $S_i$ and has 
$|H^{+}\cap S_i| =\alpha_i$, for $i \in [d]$~\cite{sz-ghs}.
In other words, $H$ simultaneously cuts off $\alpha_i$ points
from $S_i$, for $i \in [d]$.
This statement does not necessarily hold if the sets are not well-separated, 
see Figure~\ref{figure:not-ws} for an example.
\begin{figure}
\centering
\includegraphics[width=0.5\textwidth,page=1]{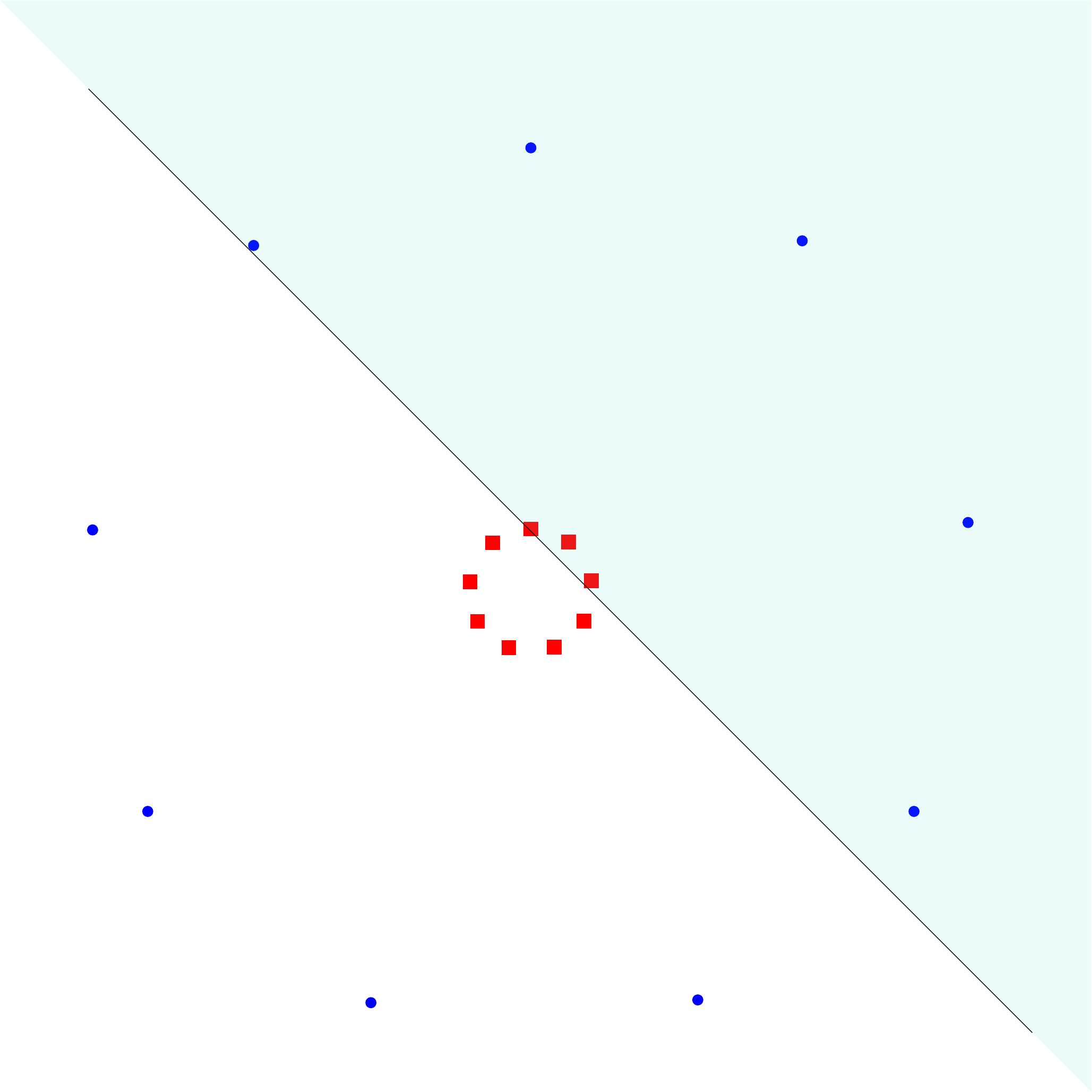}
\caption{The red (square) and the blue (round) point
sets are not well-separated. 
Every halfplane that contains three red points must contain at least five blue points.
Thus, there is no halfplane that contains exactly three red and 
three blue points.}
\label{figure:not-ws}
\end{figure}

Steiger and Zhao called their result the 
\emph{Generalized Ham-Sandwich Theorem},
yet it is not a strict generalization of the classic Ham-Sandwich Theorem.
Their result requires that the point sets obey well-separation and weak 
general position,
while the classic theorem always holds without these assumptions.
Therefore, we call this result the \emph{$\alpha$-Ham-Sandwich theorem}, for 
a clearer distinction.
Set $n = \sum_{i \in [d]} |S_i|$.
Steiger and Zhao gave an algorithm that computes the dividing hyperplane 
in $O\left(n(\log n)^{d-3}\right)$ time, which is exponential in $d$.
Later, Bereg~\cite{bereg-algo} improved this 
algorithm to achieve a running time of $n2^{O(d)}$, which is linear in $n$ but
still exponential in $d$.
We denote the associated computational search problem of finding 
the dividing hyperplane as \GHS.  

No polynomial algorithms are known for \HS and for \GHS if 
the dimension is not fixed, and the notion of approximation 
is also not well-explored. 
Despite their superficial similarity, it is not immediately
apparent whether the two problems are comparable in terms of their complexity.
Due to the additional requirements on an input for \GHS, an instance of
\HS may not be reducible to \GHS in general.

Since a dividing hyperplane for \GHS is guaranteed to exist 
if the sets satisfy the conditions of well-separation
and (weak) general position, \GHS is a total search problem.
In general, such problems are modelled by the complexity class 
\TFNP (Total Function Nondeterministic Polynomial)
of \NP-search problems that always admit a solution.
Two popular subclasses of \TFNP, originally defined by 
Papadimitriou~\cite{papadi-main},
are \PPA (Polynomial Parity Argument) its sub-class \PPAD. 
These classes contain total search problems where the
existence of a solution is based on a parity argument in 
an undirected or in a directed graph, respectively.
Another sub-class of \TFNP is \PLS (polynomial local search). 
It models total search problems where the solutions can be obtained as
minima in a local search process, while the number of 
steps in the local search may be exponential in the input size.
The class \PLS was introduced by Johnson, Papadimitriou, and 
Yannakakis~\cite{jpy-pls}.
A noteworthy sub-class of $\PPAD \cap \PLS$ is \CLS (continuous local 
search)~\cite{dp-cls}. It models similar local search problems
over a continuous domain using a continuous potential function.

Up to very recently, these complexity classes have mostly
been studied in the context of algorithmic game theory.
However, there have been increasing efforts towards mapping
the complexity landscape of existence theorems in high-dimensional
discrete geometry. Computing an approximate 
solution for the search problem associated with 
the Borsuk-Ulam Theorem is in \PPA.
In fact, this problem is complete for this class.
The discrete analogue of the Borsuk-Ulam Theorem, Tucker's Lemma~\cite{tucker}, 
is also  \PPA-complete~\cite{abb-tucker}.
Therefore, since the traditional proof of the Ham-Sandwich Theorem
goes through the Borsuk-Ulam Theorem, it follows that
\HS lies in \PPA.
In fact, Filos-Ratsikas and Goldberg~\cite{fg-hamsandwich} recently showed that 
\HS is complete for \PPA.
The (presumably smaller) class \PPAD is associated with 
fixed-point type problems: computing an approximate 
Brouwer fixed point is a prototypical complete problem for \PPAD.
The discrete analogue of Brouwer's Fixed Point Theorem, 
Sperner's Lemma, is also complete for \PPAD.
In a celebrated result, the relevance of \PPAD for
algorithmic game theory was made clear
when it turned out that computing a Nash-equilibrium in a
two player game is \PPAD-complete~\cite{cdt-nash}.
In discrete geometry, finding a solution to the Colorful 
\Caratheodory problem~\cite{barany-cc}
was shown to lie in the intersection 
$\PPAD \cap \PLS$~\cite{mmss-cc,MulzerSt18}.
This further implies that finding a \emph{Tverberg} partition (and
computing a centerpoint) also lies in the 
intersection~\cite{tverberg-original,sarkaria,dgmm-survey}.
The problem of computing the (unique) fixed point of a contraction 
map is known to lie in \CLS\cite{dp-cls}.

Recently, at ICALP~2019, Fearley, Gordon, Mehta, and Savani defined a sub-class 
of \CLS that represents a family of total search problems with 
unique solutions~\cite{fgms-ueopl}.
They named the class \emph{Unique End of Potential Line} (\cUEOPL) and 
defined it through the canonical complete problem \UEOPL.
This problem is modelled as a directed graph.
There are polynomially-sized Boolean circuits that compute
the successor and predecessor of each node, 
and a potential value that always increases on a directed path.
There is supposed to be only a single vertex with no predecessor (\emph{start of line}).
Under these conditions, there is a unique path in the graph that
ends on a vertex (called \emph{end of line}) with the highest potential along the path.
This vertex is the solution to \UEOPL.
Since the uniqueness of the solution is guaranteed only under certain assumptions,
such a formulation is called a \emph{promise} problem.
Since there seems to be no efficient way to verify the assumptions,
the authors allow two possible outcomes of the
search algorithm: 
either report a correct solution, or provide any solution that 
was found to be in violation of 
the assumptions. This formulation turns \UEOPL into a \emph{non-promise} problem and places it in \TFNP, since a correct solution is bound
to exist when there are no violations, and otherwise a violation can 
be reported as a solution.
Fearley et al.~\cite{fgms-ueopl} also introduced the concept of a \emph{promise-preserving} reduction between two problems $A$ and $B$, 
such that if an instance of $A$ has no violations, then the reduced instance of $B$ is also free of violations.
This notion is particularly meaningful for non-promise problems.

\subparagraph*{Contributions.}
We provide the first non-trivial containment in 
a complexity class for the
$\alpha$-Ham-Sandwich problem by locating it in \cUEOPL.
More precisely, we formulate \GHS as a non-promise problem
in which we allow for both valid solutions representing
the correct dividing hyperplane, as well as violations 
accounting for the lack of well-separation 
and/or (weak) general position of the input point sets.
A precise formulation of the problem is given in Definition~\ref{def:GHS} in Section~\ref{section:prelim}.
We then show a promise-preserving reduction from \GHS 
to \UEOPL.
This implies that \GHS lies in \cUEOPL, and hence in $\CLS \subseteq \PPAD \cap \PLS$.
See Figure~\ref{figure:classes} for a pictorial description.

\begin{figure}
\centering
\includegraphics[width=0.6\textwidth]{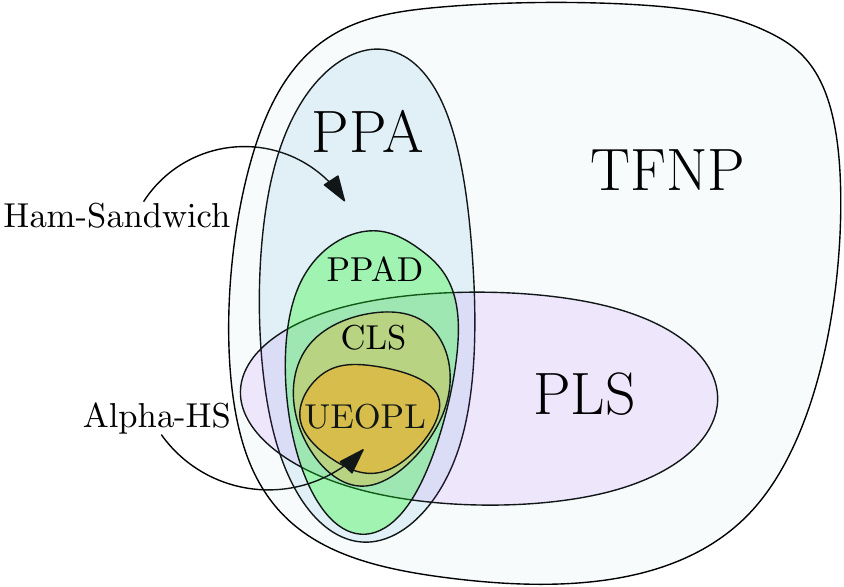}
\caption{The hierarchy of complexity classes.}
\label{figure:classes}
\end{figure}
It is not surprising to discover that \GHS lies in 
\PPAD, since the proof of the continuous version in~\cite{bhj-ghs}
was based on Brouwer's Fixed Point Theorem.
The observation that it also lies in \PLS is new and noteworthy, 
putting \GHS into the reach of local search algorithms.
In contrast, given our current understanding of total
search problems, it is unlikely that the
problem \HS would be in \PLS.

Since \GHS lies in $\PPAD \subseteq \PPA$, it is 
computationally easier than \HS, which is \PPA-complete.
This implies the existence of a polynomial-time reduction from \GHS to \HS.
A reduction in the other direction is unlikely.
It thus turns out that
well-separation brings down the complexity
of the problem by a significant amount.

Often, problems in \TFNP come in the guise of a 
polynomial-size Boolean circuit with some property. 
In contrast, \GHS is a purely geometric problem
that has no circuit in its problem definition.
This is the second problem in \cUEOPL apart from the $P$-Matrix Linear 
complementarity problem and one of the few in \CLS that does not have a description
in terms of circuits.

Our local-search formulation is based 
on the intuition of rotating a hyperplane 
until we reach the desired solution.
We essentially start with a hyperplane that is tangent 
to the convex hull of each input set, and we deterministically
rotate the hyperplane until it hits a new point.
This rotation can be continued whenever the hyperplane hits a new point,
until we reach the correct dividing hyperplane.
In other words, we can follow a local-search argument to find the solution.
We show that this sequence of rotations can be modelled 
as a canonical path in a grid graph,
and we give a potential function that guides the rotation and 
always increases along this path.
Every violation of well-separation and (weak) general position 
can destroy this path.
Furthermore, no efficient methods to verify these two
assumptions are known.
This poses a major challenge in handling the violations.
One of our main technical contributions is to handle the violation
solutions concisely.

An alternative approach would have been to look at the dual space of points
where we get an arrangement of hyperplanes.
The dividing hyperplane could then be found by looking 
at the correct level sets of the arrangement.
However, this approach has the problem that the orientations 
of the hyperplanes in the original 
space and the dual space are not consistent.
This complicates the arguments on the level sets, so
we found it more convenient to use our notion of rotating hyperplanes.
We show that we can maintain a consistent orientation 
throughout the rotation, and
an inconsistent rotation is detected as a violation
of the promise.

\subparagraph*{Outline of the paper.}
We discuss the background about the
$\alpha$-Ham-sandwich Theorem and \UEOPL in Section~\ref{section:prelim}.
In Section~\ref{section:ghs-main}, we describe our instance of \GHS
and give an overview of the reduction and violation-handling.
The technical details of the reduction are presented in Section~\ref{section:double-wedge} and Section~\ref{section:technical}.
We conclude in Section~\ref{section:conclusion}.

\section{Preliminaries}
\label{section:prelim}

\subsection{The \texorpdfstring{$\bm{\alpha}$-}{a-}Ham-Sandwich problem}
\label{subsection:ghs-prelim}

For conciseness, we describe the discrete version of $\alpha$-Ham-Sandwich Theorem~\cite{sz-ghs} here.
The continuous version~\cite{bhj-ghs} follows a similar formulation.

Let $P_1,\dots,P_d\subset \R^d$ be a collection of $d$ finite point sets.
Let $n_1,\dots,n_d$ denote the sizes of $P_1,\dots,P_d$, respectively.
For each $i\in [d]$ we say that the point set $P_i$ represents a unique color and let $P:=P_1\cup\dots\cup P_d$ denote the union
of all the points.
A set of points $\{p_1,\dots,p_m\}$ is said to be \emph{colorful} if there are no two
points $p_i,p_j$ both from the same color.
Indeed a colorful point set can have size at most $d$.

\subparagraph*{Weak general position.} 
We say that $P$ has \emph{very weak general position}~\cite{sz-ghs}, if 
for every choice of points $x_1 \in P_1, \dots, x_d\in P_d$, the affine hull of the set $\{x_1,\dots,x_d\}$ is a $(d-1)$-flat
and does not contain any other point of $P$.
This definition is sufficient for the result of Steiger and Zhao, where they simply call it as weak general position.
Of course, this definition of weak general position has no restriction on sets $\{x_1,\dots,x_d\}$
that contain multiple points from the same color.
To simplify our proofs we need a slightly stronger form of general position.
We say that $P$ has \emph{weak general position} if the above restriction also applies to
sets having exactly $d-1$ colors.
That means, each color may contribute at most one point to the set, except perhaps one color which is allowed to contribute two points.
A certificate for checking violations of weak general position is a set of $d+1$ points whose affine hull has dimension at most $d-1$, 
with at least $d-1$ colors in the set.
Testing whether a planar point set is in general position can be shown to be \NP-Hard,
using the result in~\cite{fknn-generalpos}. 
It is easy to see that when $d=2$, weak general position is equivalent to general position.

\subparagraph*{Well-separation.} 
The point set $P$ is said to be \emph{well-separated}~\cite{sz-ghs,bhj-ghs}, if 
for every choice of points $y_1\in \conv(P_{i_1}),\dots,y_k\in \conv(P_{i_k})$, where $i_1,\dots,i_k$ are distinct indices and $1\le k\le d+1$,
the affine hull of $\{ y_1,\dots,y_k \}$ is a $(k-1)$-flat.
An equivalent definition is as follows:
$P$ is well-separated if and only if for every disjoint pair of index sets $I,J\subset [d]$, there is hyperplane
that separates the set $\{ \cup_{ i\in I} P_i \} $ from the set $\{ \cup_{j\in J} P_j \}$ strictly.
Formally:

\begin{lemma}
\label{lemma:ws-formats}
Let $y_1,\dots,y_d$ be a colorful set of points in the corresponding $\conv(P_i)$.
The affine hull of $y_1,\dots,y_d$ has dimension $d-2$ or less if and only if there is a partition of $[d]$ into 
index sets $I,J$
such that $\conv\left(\{ \cup_{ i\in I} P_i \}\right) \cap \conv\left(\{ \cup_{ j\in J} P_j \} \right) \neq \emptyset$.

Given such a colorful set, the partition of $[d]$ can be computed in $\poly(n,d)$ time.
Vice-versa, given such a partition, the colorful set can be computed in $\poly(n,d)$ time.
\end{lemma}
\begin{proof}

First we prove the reverse implication: we are given $y_1,\dots,y_d$, and the affine hull has dimension at most $d-2$.
By Radon's theorem~\cite{radon} there is a partition of $y_1,\dots,y_d$ into two sets $\{ y_{i_1},\dots,y_{i_{m}} \}$ 
and $\{ y_{j_1},\dots,y_{j_{d-m}} \}$ such that their convex hulls intersect in some point $z\in \R^d$.
Then for the sets $I=\{i_1,\dots,i_{m}\},J=\{j_1,\dots,j_{d-m}\}$, we have that 
$z\in \conv(\{ \cup_{ i\in I} P_i \}) \cap \conv(\{ \cup_{ j\in J} P_j \} )$.
Furthermore, the Radon partition (and hence the partition of $[d]$) can be computed in $O(d^{3})$ time by 
solving a system of linear equations~\cite{mt-lectures}.
This proves the first part of the claim.

For the other direction, we first use linear programming to find a point $z$ in the intersection of 
$\conv(\{ \cup_{ i\in I} P_i \})$ and $\conv(\{ \cup_{ j\in J} P_j \} )$. 
Using \Caratheodory's Theorem~\cite{caratheodory}, there exists a subset $Q=\{p_1,\ldots,p_{d+1}\}\subset \{ \cup_{ i\in I} P_i \}$
of $d+1$ points such that $z\in \conv(Q)$.
Let $c_1$ be the number of colors in $Q$. 
We shrink $Q$ so that it contains one point from the convex hull of each color.
Since $\conv(Q)$ is a $d$-simplex, $\conv(Q)$ contains an edge for every pair of points $(p_i, p_j)$. 
For $(p_i,p_j)$ with the same color $t$, we shrink the edge $(p_i,p_j)$ to a point $x_t \in \conv(P_t)$ such that 
$z$ stays inside $\conv\left(Q\cup\{x_t\}\setminus\{p_i,p_j\}\right)$.
In this process $Q:=Q\cup\{x_t\}\setminus\{p_i,p_j\}$ shrinks to a $(d-1)$-simplex. 
We repeat this process until all points of $Q$ with color $t$ are shrunk to a single point.
We continue this process for the remaining colors, ending at a simplex $S_1$ of dimension $c_1-1$.
We apply the same procedure for $z \in \conv(\{ \cup_{ j\in J} P_j \})$ to obtain another simplex $S_2$ of dimension $c_2-1$. 
Since $z \in S_1 \cap S_2$, the lowest dimension of a flat containing $S_1$ and $S_2$ is at most $c_1 + c_2 -2\le d-2$. 
For each color not in $S_1$ and $S_2$, we select an arbitrary point for each, 
then $S_1$, $S_2$ and the chosen points span a $(d-2)$-flat.
$Q$ can be computed in polynomial time~\cite{dgmm-survey,mt-lectures} along with the other steps in the construction.
Therefore, the colorful set can be computed in polynomial time, proving our claim.
\end{proof}

A certificate for checking violations of well-separation is a colorful set $\{ x_1,\dots,x_d \}$ whose affine hull
has dimension at most $d-2$.
Another certificate is a partition $I,J\subset[d]$ such that the convex hulls of the indexed sets are 
not separable.
Due to Lemma~\ref{lemma:ws-formats}, both certificates are equivalent and either can be converted to the other
in polynomial time.
To the best of our knowledge, the complexity of testing well-separation
is unknown.

Given any set of positive integers $\{ \alpha_1,\dots,\alpha_d \}$ satisfying $1\le \alpha_i\le n_i$, $i\in[d]$,
an \emph{$(\alpha_1,\dots,\alpha_d)$-cut} is an oriented hyperplane $H$ that contains one point from each color
and satisfies $|H^{+}\cap P_i|=\alpha_i$ for $i\in[d]$, where $H^{+}$ is the closed positive half-space defined by $H$.

\begin{theorem}[$\alpha$-Ham-Sandwich Theorem~\cite{sz-ghs}]
\label{thm:ghs}
Let $P_1, \ldots, P_d$ be finite, well-separated point sets in $\R^d$. 
Let $\alpha=(\alpha_1, \dots, \alpha_d)$ be a vector, where $\alpha_i \in [n_i]$ for $i\in [d]$.
\begin{enumerate}
\item If an $\alpha$-cut exists, then it is unique.

\item If $P$ has weak general position, then a cut exists
for each choice of $\alpha$, $\alpha_i \in [n_i]$.
\end{enumerate}
\end{theorem}

That means, every colorful $d$-tuple of $P$ corresponds to exactly one \av.
Steiger and Zhao~\cite{sz-ghs} also presented an algorithm to compute the cut in $O(n(\log n)^{d-3})$ time, where $n=\sum_{i=1}^{d} n_i$.
The algorithm proceeds inductively in dimension and employs a prune-and-search technique.
Bereg~\cite{bereg-algo} improved the pruning step to improve the runtime to $n2^{O(d)}$.

\subsection{Unique End of Potential Line}
\label{subsection:ueopl-prelim}

We briefly explain the \emph{Unique end of potential line} problem that was introduced in~\cite{fgms-ueopl}.
More details about the problem and the associated class can be found in the above reference.
\begin{definition}[from\cite{fgms-ueopl}]
\label{def:UEOPL}
Let $n,m$ be positive integers.
The input consists of
\begin{itemize}
\item a pair of Boolean circuits $\S, \P : \{0,1\}^n \rightarrow \{0,1\}^n$ such that $\P(0^n) = 0^n \neq \S(0^n)$, and 

\item a Boolean circuit $\V : \{0,1\}^n \rightarrow \{0,1,\ldots,2^m-1\}$ such that $\V(0^n)=0$,
\end{itemize}
each circuit having $\poly(n,m)$ size.
The \UEOPL problem is to report one of the following:
\renewcommand{\theenumi}{\bf(U\arabic{enumi})}
\begin{enumerate}
\setlength{\itemindent}{0.9cm}
\item \label{sol:U1}  A point $v \in \{0,1\}^n$ such that $\P(\S(v)) \neq v$. 
\end{enumerate}	
\vspace{-0.125in}
\renewcommand{\theenumi}{\bf(UV\arabic{enumi})}
\begin{enumerate}
\setlength{\itemindent}{0.9cm}
\item \label{sol:UV1} A point $v \in \{0,1\}^n$ such that $\S(v) \neq v$, $\P(\S(v)) = v$, and $\V(\S(v))-\V(v) \leq 0$.
\item \label{sol:UV2} A point $v \in \{0,1\}^n$ such that $\S(\P(v)) \neq v \neq 0^n$. 	
\item \label{sol:UV3} Two points $v,u \in \{0,1\}^n$ such that $v\neq u$, $\S(v) \neq v$, $\S(u) \neq u$, 
and either $\V(v) = \V(u)$ or $\V(v) < \V(u) < \V(\S(v))$.
\end{enumerate}
\end{definition}
The problem defines a graph $G$ with up to $2^{n}$ vertices.
Informally, $\S(\cdot),\P(\cdot),\V(\cdot)$ represent the \emph{successor}, \emph{predecessor} and \emph{potential} functions that act on each vertex in $G$.
The in-degree and out-degree of each vertex is at most one.
There is an edge from vertex $u$ to vertex $v$ if and only if $\S(u)=v$, $\P(v)=u$ and $\V(u)<\V(v)$.
Thus, $G$ is a directed acyclic path graph (line) along which the potential strictly increases.
The condition $\S(\P(x))\neq x$ means that $x$ is the start of the line, 
$\P(\S(x))\neq x$ means that $x$ is the end of the line, and $\P(\S(x))= x$ occurs when $x$ is neither.
The vertex $0^{n}$ is a given start of the line in $G$.

\ref{sol:U1} is a solution representing the end of a line.
\ref{sol:UV1}, \ref{sol:UV2} and \ref{sol:UV3} are violations.
\ref{sol:UV1} gives a vertex $v$ that is not the end of line, and the potential of $\S(v)$ is not strictly larger than that of $v$, which is a violation
of our assumption that the potential increases strictly along the line.
\ref{sol:UV2} gives a vertex that is the start of a line, but is not $0^{n}$.
\ref{sol:UV3} shows that $G$ has more than one line, which is witnessed by the fact that $v$ and $u$ cannot lie on the same line if they
have the same potential, or if the potential of $u$ is sandwiched between that of $v$ and the successor of $v$.
Under the promise that there are no violations, $G$ is a single line starting at $0^{n}$ and ending at a vertex that is the unique solution.
\UEOPL is formulated in the non-promise setting, placing it in the class \TFNP.

The complexity class \cUEOPL represents the class of problems that can be reduced in polynomial time to \UEOPL.
This has been shown to lie in \CLS in~\cite{fgms-ueopl} and contains three classical problems: finding the fixed point of a contraction map,
solving the P-Matrix Linear complementarity problem, and finding the unique sink of a directed graph (with arbitrary edge orientations) on the 
1-skeleton of a hypercube.

A notion of \emph{promise-preserving} reductions is also defined in~\cite{fgms-ueopl}.
Let $X$ and $Y$ be two problems both having a formulation that allows for valid and violation solutions.
A reduction from $X$ to $Y$ is said to be promise-preserving, if whenever it is promised that $X$ has no violations, then the reduced instance of $Y$ also
has no violations.
Thus a promise-preserving reduction to \UEOPL would mean that whenever the original problem is free of violations, then the reduced instance
always has a single line that ends at a valid solution.

\subsection{Formulating the search problem}
\label{subsection:ghs-formal}

We formalize the search problem for $\alpha$-Ham-Sandwich in a non-promise setting:
\begin{definition}[Alpha-HS]
\label{def:GHS}
Given $d$ finite sets of points $P = P_1 \cup \ldots \cup P_d$ in $\R^d$ and a vector $(\alpha_1, \ldots, \alpha_d)$ of positive integers such that 
$\alpha_i \leq |P_i|$ for all $i\in[d]$, the \GHS problem is to find one of the following:
\renewcommand{\theenumi}{\bf(G\arabic{enumi})}
\begin{enumerate}
	\setlength{\itemindent}{0.9cm}
\item \label{sol:G1} An $(\alpha_1,\ldots,\alpha_d)$-cut.
\end{enumerate}
\vspace{-0.125in}
\renewcommand{\theenumi}{\bf(GV\arabic{enumi})}
\begin{enumerate}
	\setlength{\itemindent}{0.9cm}
\item \label{sol:GV1} A subset of $P$ of size $d+1$ and at least $d-1$ colors that lies on a hyperplane.
\item \label{sol:GV2} A disjoint pair of sets $I,J \subset [d]$ such that $\conv(\{ \cup_{ i\in I} P_i \}) \cap \conv(\{ \cup_{ j\in J} P_j \} ) \neq \emptyset$.
\end{enumerate}
\end{definition}
Here a solution of type~\ref{sol:G1} corresponds to a solution representing a valid cut,
while solutions of type \ref{sol:GV1} and \ref{sol:GV2} refer to violations of weak general position and well-separation, respectively.
From Theorem~\ref{thm:ghs} we see that a valid solution is guaranteed if no violations are presented, which shows that \GHS is a total search problem.

\section{Alpha-HS is in UEOPL}
\label{section:ghs-main}

In this section we describe our instance of \GHS in more detail and briefly outline a reduction to \UEOPL.

\subparagraph*{Setup.}
The input consists of $d$ finite point sets $P_1,\dots,P_d\subset \R^d$ each representing a unique color, 
of sizes $n_1,\dots,n_d$, respectively, 
and a vector of integers $\alpha=(\alpha_1,\dots,\alpha_d)$ such that $\alpha_i \in[n_i]$ for each $i\in[d]$.
Let $k$ denote the number of coordinates of $\alpha$ that are not equal to one.
Without loss of generality, we assume that $\{ \alpha_1,\dots,\alpha_k \}$ are the non-unit entries in $\alpha$.
Let $P$ denote the union $P_1\cup\dots\cup P_d$.
For each $i\in[d]$ we define an arbitrary order $\prec_i$ on $P_i$.
Concatenating the orders $\prec_1,\prec_2,\dots,\prec_{d}$ in sequence
gives a global order $\prec$ on $P$.
That means, $p\prec q$ if $p\in P_i, q\in P_j$ and $i<j$ or 
$p,q\in P_j$ and $p\prec_j q$.

We follow the notation of~\cite{sz-ghs} to define the orientation of a 
hyperplane in $\R^d$ that has a non-empty intersection with each convex hull of $P_i$.
For any hyperplane $H$ passing via $\{ x_1\in \conv(P_1),\dots, x_d\in  \conv(P_d)\}$,
the normal is the unit vector $\hat{n}\in \R^d$ that satisfies
$\dotp{x_i}{\hat{n}}=t$ for some fixed $t\in\R$ and each $i\in[d]$, and
\[
\mathrm{det}
\begin{vmatrix}
x_1 & x_2 & \dots & x_d & \hat{n} \\
1 & 1 & \dots & 1 & 0
\end{vmatrix}
>0.
\]
The positive and negative half-spaces of $H$ are defined accordingly.
In~\cite[Proposition\,2]{bhj-ghs}, the authors show that the choice of $\hat{n}$ does not depend on the choice of $x_i\in \conv(P_i)$ for any $i$,
if the colors are well-separated. 
Notice that if the colors are not well-separated, then the dimension of the affine hull of $\{x_1,\ldots,x_d\}$ may be less than $d-1$.
This makes the value of the determinant above to be zero, so the orientation is not well-defined.

We call a hyperplane \emph{colorful} if it passes through exactly $d$ colorful points $p_1,\dots,p_d\subset P$.
Otherwise, we call the hyperplane \emph{non-colorful}.
There is a natural orientation for colorful hyperplanes using the definition above.
In order to define an orientation for non-colorful hyperplanes, one needs additional points from the convex hulls of unused colors 
on the hyperplane.
Let $H'$ denote a hyperplane that passes through points of $(d-1)$ colors.
Let $P_j$ denote the missing color in $H'$.
To define an orientation for $H'$, we choose a point from $\conv(P_j)$ that lies on $H'$ as follows.
We collect the points of $P_j$ on each side of $H'$, and choose the highest ranked points under the order $\prec_j$.
Let these points on opposite sides of $H'$ be denoted by $x$ and $y$.
Let $z$ denote the intersection of the line segment $xy$ with $H'$.
By convexity, $z$ is a point in $\conv(P_j)$, so we choose
$z$ to define the orientation of $H'$.
The intersection point $z$ does not change if $x$ and $y$ are interchanged,
giving a valid definition of orientation for $H'$.
We can also extend this construction to define orientations for
hyperplanes containing points from less than $d-1$ colors, 
but for our purpose this definition suffices.
The \emph{\av} of any oriented hyperplane $H$ is a $d$-tuple $(\alpha_1,\dots,\alpha_d)$ of integers
where $\alpha_i$ is the number of points of $P_i$ in the closed halfspace $H^{+}$ for $i\in [d]$.

\subsection{An overview of the reduction}
\label{subsection:reduction-overview}
We give a short overview of the ideas used in the reduction from \GHS to \UEOPL.
The details are technical and we defer them to Section~\ref{section:technical}.
We encourage the interested reader to go through the details of our reduction.

Our intuition is based on rotating a colorful hyperplane $H$ to another colorful hyperplane $H'$ through a sequence of local changes 
of the points on the hyperplanes such that the \av of $H'$ increases in some coordinate by one from that of $H$. 
We next define the rotation operation in a little more detail.
An \emph{anchor} is a colorful $(d-1)$-tuple of $P$ which spans a $(d-2)$-flat.
The following procedure takes as input an anchor $R$ and some point $p\in P\setminus R$
and determines the next hyperplane obtained by a rotation. 
The output is $(R',p')$, where $R'$ is an anchor and $p'\in P\setminus R'$ is some point.

\vspace{0.15in}
\noindent{\bf{Procedure}} $(R',p') = NextRotate(R,p)$
\begin{enumerate}
	\item Let $H$ denote the hyperplane defined by $R \cup \{p\}$ and $t_1$ be the missing color in $R$.
	\item If the orientation of $H$ is not well-defined, report a violation of weak general position and well-separation.
	\item Let $P^+_{t_1}$ be the subset of $P_{t_1}$ that lies in the closed halfspace $H^+$ and $P^-_{t_1}$ be the subset of $P_{t_1}$ that lies in the open halfspace $H^-$. Let $x \in P^+_{t_1}$ be the highest ranked point according to the order $\prec_{t_1}$ and $y \in P^-_{t_1}$ be the highest ranked point according to $\prec_{t_1}$. 
	\item If $p$ has color $t_1$ and $|P^+_{t_1}| = n_{t_1}$, report out of range.
	\item We rotate $H$ around the anchor $R$ in a direction such that the hyperplane is moving away from $x$ along the segment $xy$ until it hits some point $q \in P$.
	\item If the hyperplane hits multiple points at the same time, report a violation of weak general position.
	\item If $p'$ is not color $t_1$, set $R' := R \cup \{q\} \setminus \{r\}$ and $p'=r$, where $r$ is a point in $R$ with the same color as $p'$. Otherwise, set $R' = R$ and $p'=q$.
	\item \textsf{Return} $(R',p')$.
\end{enumerate}
Figure~\ref{figure:why-double} shows an application of this procedure, rotating $H_0$ to $H_4$ through $H_1,H_2,H_3$.
\begin{figure}
	\centering
	\hspace{0.5cm}
	\includegraphics[width=0.85\textwidth]{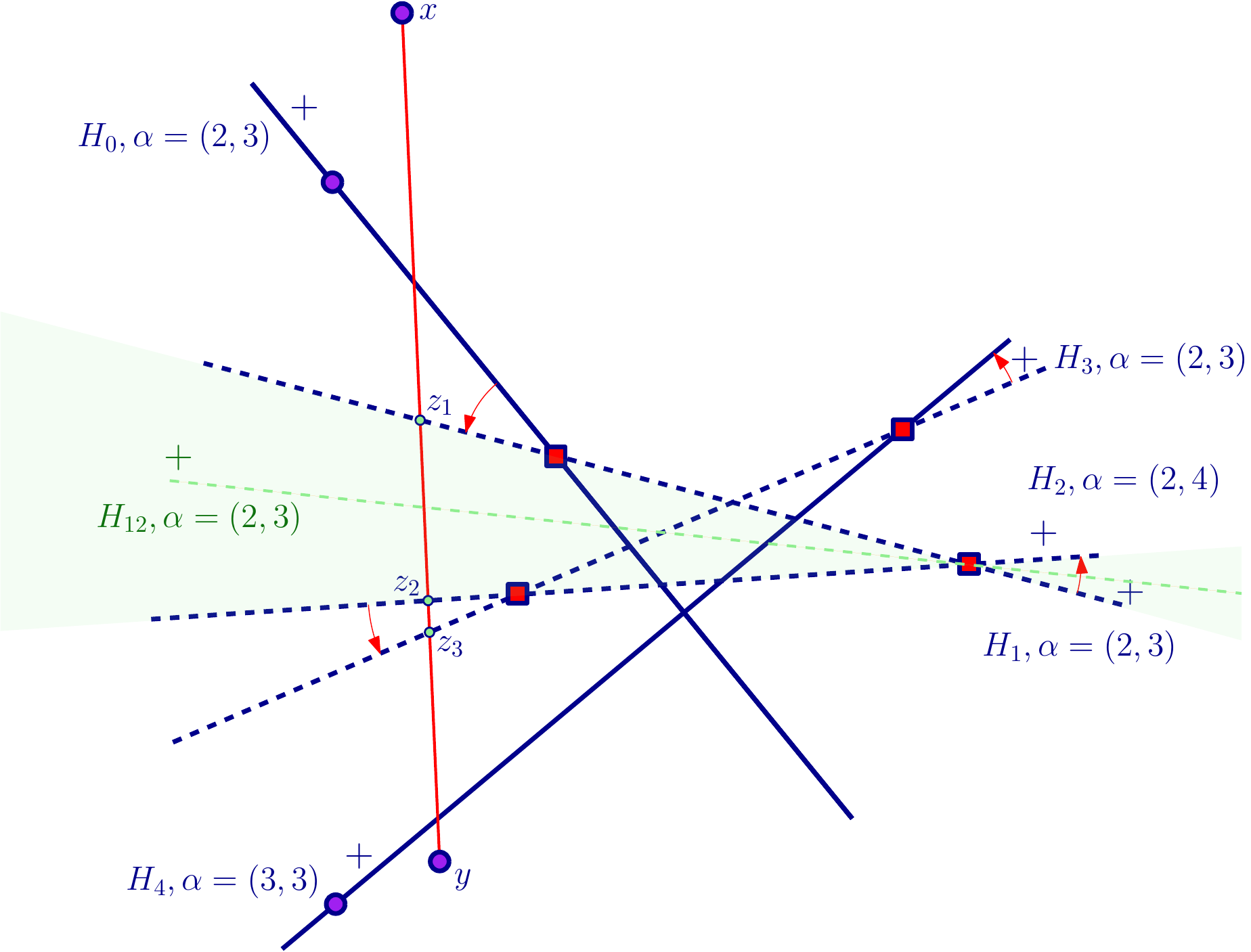}
	\caption{An example showing a sequence of rotations from $H_0$ to $H_4$ through $H_1,H_2,H_3$. Purple (disk) is the first color and red (square) is the second color. 
	This sequence represents a path between two vertices in the \UEOPL graph that is generated in the reduction.
	The double-wedge is shaded and its angular bisector $H_{12}$ has the desired \av.}
	\label{figure:why-double}
\end{figure}

This rotation function can be interpreted as a function that assigns each hyperplane to the next hyperplane.
The set of colorful hyperplanes can be interpreted as vertices in a graph with the rotation function determining the connectivity of the graph.

\subparagraph*{Canonical path.} 
Each colorful hyperplane $H$ is incident to a colorful set of $d$ points.
This set of points defines $d$ possible anchors, and each anchor can be used to rotate $H$ in a different fashion.
To define a unique sequence of rotations, we pick a specific order as follows:
first, we assume that the colorful hyperplane $H$ whose \av is $(1,\ldots,1)$ is given (we show later how this assumption can be removed). We start at $H$ and pick the anchor that excludes the first color, then apply a sequence of rotations until we hit another colorful hyperplane with \av $(2,1,\ldots, 1)$. Similarly, we move to a colorful hyperplane with \av $(3,1,\ldots,1)$ and so on until we reach $(\alpha_1,1,\ldots,1)$. Then, we repeat this for the other colors in order to reach $(\alpha_1,\alpha_2,1,\ldots,1)$ and so on until we reach the target \av. 
This pattern of $\alpha$-vectors helps in defining a potential function that strictly increases along the path.
We can encode this sequence of rotations as a unique path in the \UEOPL instance, and we call it \emph{canonical path}. 

A natural way to define the \UEOPL graph would be to consider hyperplanes as the vertices in the graph.
However, this leads to complications. 
Figure~\ref{figure:why-double} shows a rotation from $H_0$ to $H_4$, with {\av}s $(2,3)$ and $(3,3)$ respectively.
During the rotation, we encounter a hyperplane $H_2$ for which its \av is $(2,4)$, which differs from our desired sequence of $(2,3),\dots,(2,3),(3,3)$.
This makes it difficult to define a potential function in the graph that strictly increases along the path $v_{H_0},\dots,v_{H_4}$ where $v_{H_i}$
is the vertex representing hyperplane $H_i$.
One way to alleviate this problem is to not use $H_i$ as a vertex directly, but the \emph{double-wedge} that is traced out by the rotation from $H_{i}$
to $H_{i+1}$.
If the \av is now measured using the hyperplane that bisects the double-wedge, then we get the desired sequence of $(2,3),\dots,(2,3),(3,3)$.
See Figure~\ref{figure:why-double} for an example.

With additional overhead, the rotation function can be extended to double-wedges. 
This in turn also leads to a neighborhood graph where the vertices are the double-wedges and the rotations can be used to define the edges.
The graph is connected and has a grid-like structure that may be of independent interest. 
To simplify the exposition, we postpone the description of double-wedges and the associated graph to Section~\ref{section:double-wedge}.

\subparagraph*{Distance parameter and potential function.} 
The \av is not sufficient to define the potential function, since the sequence of rotations between two colorful hyperplanes may have the same \av. 
For instance, the bisectors of the rotations in $H_0,\dots,H_3$ in Figure~\ref{figure:why-double} all have the same \av.
Hence, we need an additional measurement in order to determine the direction of rotation that increases the \av. 

Similar to how we define the orientation for a non-colorful hyperplane, let $H$ denote a hyperplane that passes through points of $(d-1)$ colors. 
Let $P_j$ denote the missing color in $H$. 
Let $x,y\in P_j$ be the highest ranked points under $\prec_j$ in $H^+$ and $H^-$ respectively. Let $z$ denote the intersection of $xy$ and $H$. We define a distance parameter called \emph{\dv} of $H$ to be the distance $\|x-z\|$. 
In Figure~\ref{figure:why-double}, we can see that rotating from $H_0$ to $H_4$ sweeps the segment $xy$ in one direction, with the \dv of the hyperplanes increasing strictly. 
This is sufficient to break ties and hence determine the correct direction of rotation.
The precise statement is given in Lemma~\ref{lemma:d-value-strictly}. 
We can extend this definition to the domain of double-wedges.
We define a potential value for each vertex on the canonical path in \UEOPL using the sum of weighed components of \av and \dv for the tie-breaker. 

\subparagraph*{Correctness.} We show that if there are no violations, we can always apply {\bf{Procedure}} $NextRotate$ to increment the \av until we find the desired solution, which implies that the canonical path exists. If the input satisfies weak general position, we can see that the rotating hyperplane always hits a unique point in Step $5$, which may be swapped to form a new anchor in Step $7$.

The well-separation condition guarantees that the potential function always increases along the rotation.
Let $H_1,H_2$ denote a pair of hyperplanes that are the input and output of {\bf{Procedure}} $NextRotate$ respectively. Let $H$ denote any intermediate hyperplane during the rotation from $H_1$ to $H_2$ through the common anchor. Let $P_j$ be the color missing from the anchor and $x$ be the highest ranked point under $\prec_j$ in $H_1^+$.
We say that the orientation of $H_2$ (resp. $H$) is \emph{consistent} with that of $H_1$ if $x \in H_2^+$ (resp. $x \in H^+$). Lemma~\ref{lemma:consistent-orientation-0} shows that the orientations are always consistent when $H_1$ and $H_2$ are non-colorful hyperplanes even without the assumption of well-separation.

\begin{lemma}[consistency of orientation]
	\label{lemma:consistent-orientation-0}
	Assume that weak general position holds.
	Let $H_1,H_2$ be the input and output of {\bf{Procedure}} $NextRotate$ respectively. Let $H$ denote any intermediate hyperplane within the rotation. 
	The orientations of $H_1$ (resp. $H_2$) and $H$ are consistent when $H_1$ (resp. $H_2$) is a non-colorful hyperplane.
\end{lemma}

\begin{proof}
	Since $H_1$ is a non-colorful hyperplane, let $P_j$ denote the color missing from $H_1$. $H_1$ and $H$ give the same partition of $P_j$ into two sets because the continuous rotation from $H_1$ to $H$ does not hit any point in $P_j$. Let $x$ and $y$ be the highest ranked points under $\prec_j$ in each set. 
	Since we have weak general position, the segment $xy$ cannot pass through the anchor of the rotation so that the orientations of $H_1$ and $H$ are well-defined by the $(d-1)$ colored points in the anchor and the intersections of the hyperplanes with the segment $xy$. Thus, the determinant defining the normal of the rotating hyperplane from $H_1$ to $H$ for the orientation is always non-zero. Since the intersection of the rotating hyperplane from $H_1$ to $H$ and the segment $xy$ moves continuously along $xy$, by a continuity argument, the normal of the hyperplane does not flip during the rotation. Without loss of generality, assume that $x\in H_1^{+}$. This implies that $x$ is always in the positive half-space of $H$ and hence $H$ has a consistent orientation as $H_1$. The same proof holds for $H_2$.
\end{proof}

Next, we show that the \dv is strictly increasing for all the intermediate hyperplanes in the sequence of rotations from one colorful hyperplane to another colorful hyperplane.

\begin{lemma}
	\label{lemma:d-value-strictly}
	Assume that weak general position holds.
	Let $H_0$ be a colorful hyperplane and $H_k$ be the first colorful hyperplane obtained by a sequence of rotations by {\bf{Procedure}} $NextRotate$. We denote $H_1,\ldots,H_{k-1}$ be the non-colorful hyperplanes obtained from the above sequence of rotations.
	The {\dv}s of $H_1,\ldots,H_{k-1}$ is strictly increasing.
\end{lemma}

\begin{proof}
	Let $P_j$ denote the color missing from $H_1$. Then, $H_2,\ldots,H_{k-1}$ all miss the color $P_j$, otherwise $H_k$ is not the first colorful hyperplane obtained by the rotations. Therefore, each $H_i$ gives the same partition of $P_j$ into two sets for $i=1,\ldots,k-1$ because the continuous rotations from $H_1$ to $H_{k-1}$ does not hit any point in $P_j$. Let $x$ and $y$ be the highest ranked points under $\prec_j$ in each set. Without loss of generality, assume that $x\in H_1^{+}$. Since $H_1,\ldots,H_{k-1}$ are non-colorful hyperplanes, by Lemma~\ref{lemma:consistent-orientation-0}, the consistent of the orientation can carry from $H_1$ to $H_2$ and so on. Then we have $x\in H_1^{+},\dots, x\in H_{k-1}^{+}$ and $y\in H_1^{-},\dots, y\in H_{k-1}^{-}$. Let $z_1=xy\cap H_1, \dots, z_{k-1} = xy\cap H_{k-1}$. According to Step $5$ of {\bf{Procedure}} $NextRotate$, each rotation is performed by moving away from $x$ along the segment $xy$. Hence we have $\|x-z_1\| < \|x-z_2\| < \dots < \|x-z_{k-1}\|$.
\end{proof}

The last step for proving that the potential function always increases along the canonical path is to show that the \av increases in some coordinate from one colorful hyperplane to another colorful hyperplane through {\bf{Procedure}} $NextRotate$. This requires the assumption of well-separation.
Lemma~\ref{lem:inconsistent-orientation-0} shows that if the orientations of $H_1,H_2$ and $H$ are inconsistent, then well-separation is violated. By the contrapositive, if well-separation is satisfied, then all hyperplanes in the rotation always give consistent orientations. Then, it implies that rotating from a colorful hyperplane $H_0$ to another colorful hyperplane $H_k$ through a sequence of non-colorful hyperplanes that miss color $P_j$, we have $H_0^+ \cap P_j \subset H_k^+ \cap P_j$ and $H_k$ contains one additional point in $P_j$ that is hit by the last rotation. Therefore, $\alpha_j$ is increased by $1$ and other $\alpha_i$s keep the same value because of the way we swap the point of repeated color with the one in the anchor and the direction of rotation.

\begin{lemma}
	\label{lem:inconsistent-orientation-0}
	Assume that weak general position holds.
	Let $H_1,H_2$ be the input and output of {\bf{Procedure}} $NextRotate$ respectively. Let $R$ denote the anchor of the rotation from $H_1$ to $H_2$, and $P_j$ denote the color missing from $R$. Let $H$ denote any intermediate hyperplane within the rotation.
	If the orientations of $H_1$ (resp. $H_2$) and $H$ are inconsistent, then $H_1$ (resp. $H_2$) is a colorful hyperplane and 
	we can find a colorful set $R \cup \{x'\}$ lying in a $(d-2)$-flat where $x' \in \conv(P_j)$, in $O(d^3)$ arithmetic operations.
	The set $R \cup \{x'\}$ witnesses the violation of well-separation.
\end{lemma}

\begin{proof}
	Since the orientations of $H_1$ and $H$ are inconsistent, $H_1$ must be a colorful hyperplane by Lemma~\ref{lemma:consistent-orientation-0}.
	Therefore, the point in $H_1$ that is not in the anchor is in $P_j$, denoted by $p$.
	
	Let $x$ and $y$ be the points defined in Lemma~\ref{lemma:consistent-orientation-0} such that $x,y \in P_j$, and $x$ and $y$ are on different sides of $H_1$ and $H$. 
	The $(d-2)$-flat containing $R$ separates $H_1$ and $H$ into two $(d-1)$-dimensional half-subspaces each. Let $H^+_{1,R}$ and $H^+_{R}$ be the half-subspaces intersecting with $xy$ on $H_1$ and $H$ respectively, and let us denote the intersection points by $z_p$ and $z$, respectively.
	The opposite half-subspaces are denoted by $H^-_{1,R}$ and $H^-_{R}$, respectively.
	By definition of the orientation for non-colorful hyperplanes, the orientation of $H$ is defined by $R \cup \{z\}$. Although the orientation of $H_1$ is defined by $R \cup \{p\}$, if we consider the determinant defining the orientation using $R \cup \{z_p\}$, it gives an orientation consistent with that of $H$. Therefore, it must be that $p \in H^-_{1,R}$. Then, we can see that the line segment $pz_p$ intersects the $(d-2)$-flat of $R$. We can compute $z_p$ and also the intersection point $x'$ of $p z_p$ and the $(d-2)$-flat of $R$ by solving systems of linear equations with $d$ equations and $d$ variables in $O(d^3)$ arithmetic operations. Since $x' \in \conv(P_j)$, $R \cup \{x'\}$ is a colorful set contained in the $(d-2)$-flat of $R$.
\end{proof}

In order to guarantee that there is no other path in \UEOPL apart from the canonical path, we introduce self-loops
for vertices that are not on the canonical path. 
The detailed proof is given in Lemma~\ref{lem:no_violations} that if there are no violations, 
then the reduced instance of \UEOPL only gives a~\ref{sol:U1} solution, 
which readily translates to a \ref{sol:G1} solution, so our reduction is promise-preserving, and this can be done in polynomial time.

Since we do not know the hyperplane with \av $(1,\dots,1)$ in advance,
we split the problem into two sub-problems: in the first we start with any colorful hyperplane.
We reverse the direction of the canonical path determined by the potential and construct an \GHS instance for which the vertex with \av $(1,\dots,1)$ is the solution.
In the second, we use this vertex as the input to the main \GHS instance.
If the input is free of violations, then both sub-problems give valid solutions and together they answer the original question.

\subparagraph*{Handling violations.}
The reduction maps violations of \GHS to those of the \UEOPL instance, and certificates for the violations can be recovered
from additional processing.
When a violation of weak general position is witnessed on a vertex that lies on the canonical path, 
a hyperplane incident to $d$ colors may contain additional points. 
This in turn implies that some $\alpha$-cut is missing, so that the correct solution for the target may not exist. 
In addition, the (highest-ranked) points $x,y$ from the missing color that we choose to define 
the orientation of a non-colorful hyperplane may form a segment $xy$ that passes through the $(d-2)$-flat spanned by the anchor. 
In that case the orientation of the hyperplane is not well-defined. 
In the reduction, these problematic vertices are removed from the canonical path, 
thereby creating some additional starting points and end points in the reduced instance. 
These violations can be captured by \ref{sol:U1} with a wrong \av or \ref{sol:UV2}. 
Furthermore, the hyperplanes that contains the degenerate point sets could be represented by different choices of anchors and a additional point on the plane.
Each such pair represents a vertex in the reduced instance. 
We join these vertices in the form of a cycle in the \UEOPL instance with all vertices having the same potential value, 
so that the violations can also be captured by \ref{sol:UV1} and \ref{sol:UV3}.

When a violation of well-separation is witnessed on a vertex on the canonical path, the orientations of 
the two hyperplanes paired by {\bf{Procedure}} $NextRotate$ may be inconsistent, 
which may not guarantee that the \av is incremented in one component by one (See Figure~\ref{figure:not-separated}). 
Hence, the canonical path is split into two paths that can be captured by \ref{sol:UV2}. 
Furthermore, a violation of well-separation also creates multiple colorful hyperplanes with the same \av (See Figure~\ref{figure:not-separated}, left). 
Two vertices in the \UEOPL graph with the same potential value, which could correspond to some colorful or non-colorful hyperplanes, can be reported by \ref{sol:UV3}. We show that this gives a certificate of violation of well-separation in the following lemmata, where $m_0$ is the number of bits
used to represent each coordinate of points of $P$. 

\begin{figure}
	\centering
	\hspace{0.5cm}
	\includegraphics[width=0.45\textwidth,page=1]{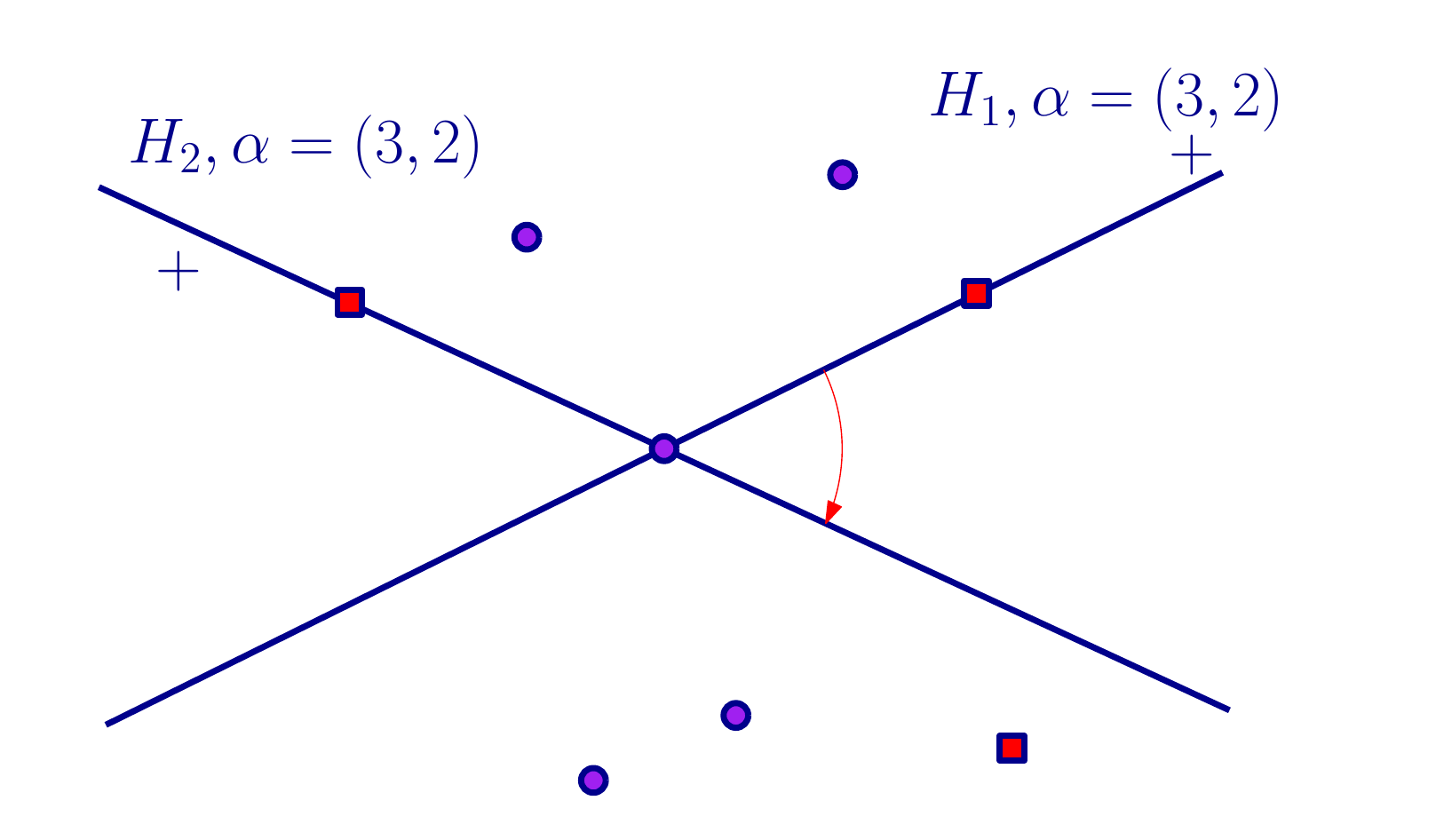}%
	\includegraphics[width=0.45\textwidth,page=2]{non-well-separated-2}
	\caption{The examples show two sets of points that are not well-separated. Purple (circle) represents the first color and red (square) represents the second color.
In both examples the rotation procedure does not increase the \av. 
Both examples show that the orientation of the hyperplane may be flipped after the rotation, so the resulting \av can go wrong.}
	\label{figure:not-separated}
\end{figure}

\begin{lemma}
\label{lem:not_separated_2}
	Given two colorful hyperplanes $H_p$, $H_q$ with the same \av, we can find a colorful set $\{x_1 \in \conv(P_1), \ldots, x_d\in \conv(P_d)\}$ that lies on a $(d-2)$-flat in $\poly(n,d,m_0)$ time.	
\end{lemma}
\begin{proof}
Let $p_1 \in P_1, \ldots, p_d \in P_d$ denote the colorful points on $H_p$ and $q_1 \in P_1, \ldots, q_d \in P_d$ denote the colorful points on $H_q$. Throughout this proof, we consider $H^+$ to be a closed halfspace while $H^-$ is an open halfspace.
We prove the claim by induction on the dimension $d$. 

For the base case $d=2$, we have three different cases to consider depending on which cells out of $H_p^- \cap H_q^+, H_p^+ \cap H_q^+, H_p^+ \cap H_q^-, H_p^- \cap H_q^-$ contain the segments $p_1q_1$ and $p_2q_2$.
If the open segments $p_1q_1$ and $p_2q_2$ lie in either $H_p^+ \cap H_q^+$ or $H_p^- \cap H_q^-$, then we can apply the same argument as in Lemma~\ref{lem:not_separated_1} to find a point $y$ that lies in $\conv(P_1) \cap \conv(P_2)$. In particular, $y$ could be the intersection point of $p_1q_1$ and $p_2q_2$ (see Figure~\ref{figure:ns2-1}).

Without loss of generality, suppose that the open segment $p_1q_1 \in H_p^- \cap H_q^+$. Since $p_1 \in H_p^+ \cap H_q^+$ and $q_1 \in H_p^-\cap H_q^+$, there exists at least one point $r_1 \in P_1 \cap (H_p^+ \cap H_q^-)$ in order for $|P_1 \cap H_p^+| = |P_1 \cap H_q^+|$ to hold. If the open segment $p_2q_2$ also lies in $H_p^- \cap H_q^+$ (resp. $H_p^+ \cap H_q^-$), then there exists at least one point $r_2$ in $P_2 \cap (H_p^+ \cap H_q^-)$ (resp. $P_2 \cap (H_p^- \cap H_q^+)$). We can see that the intersection point $x$ of $H_p$ and $H_q$ lies inside the triangles $\triangle p_1q_1r_1$ and $\triangle p_2q_2r_2$ (see Figure~\ref{figure:ns2-2}).

Suppose that the open segment $p_2q_2$ lies in $H_p^+ \cap H_q^+$ (resp. $H_p^- \cap H_q^-$). In order to assign correct orientations to $H_p$ and $H_q$,
the order in which points of $P_1$ and $P_2$ appears on the hyperplanes along any direction must be the same for both.
This is only feasible when $p_2$ lies between $p_1$ (resp. $q_1$) and the intersection point $x=H_p \cap H_q$.
Hence, $p_2$ (resp. $q_2$) lies inside the triangles $\triangle p_1q_1r_1$ (see Figure~\ref{figure:ns2-3}).

For $d > 2$, if $H_p \cap H_q$ does not intersect $\conv(P_i)$ for all $i \in [d]$, then the claim follows from Lemma~\ref{lem:not_separated_1}. Without loss of generality, suppose that $H_p \cap H_q$ intersects $\conv(P_1)$. 
We can use linear programming to check whether there is a hyperplane $H$ that separates $P_1$ and $\cup_{ i \in[2..d]}P_i$.
If $H$ does not exist, by Lemma~\ref{lemma:ws-formats}, we can find a desired colorful set and we are done. Otherwise, we use linear programming to find a point $x_1 \in H_p \cap H_q \cap \conv(P_1)$. Then, we project $\cup_{ i \in[2..d]}P_i$ towards $x_1$ onto $H$. Let $P_i'$ be the projected point set for $i\in[2..d]$. Let $H_p'$ and $H_q'$ be $H_p \cap H$ and $H_q \cap H$, respectively.
From the way we do the projection, $H_p'$ and $H_q'$ keep the same \av with respect to $P_i'$ on the hyperplane $H$. By induction, we can find a colorful set $\{x_2' \in \conv(P_2'),\ldots,x_d'\in \conv(P_d')\}$ that lies on a $(d-3)$-flat. Then, we shoot a ray from $x_1$ towards $x_i'$ until it hits some point $x_i \in \conv(P_i)$ and we can see that $\{x_1, x_2,\ldots,x_d\}$ spans a $(d-2)$-flat.
\end{proof}

\begin{lemma}
\label{lem:not_separated_3}
	Given two non-colorful hyperplanes that each contains $d-1$ points and have the same missing color, \av and \dv, 
	we can find a colorful set of points $\{x_1 \in \conv(P_1), \ldots, x_d\in \conv(P_d)\}$ that lies on a $(d-2)$-flat in $\poly(n,d,m_0)$ time.
\end{lemma}

\begin{proof}
The idea is to transform $P$ to a point set $P'$, in which we can find two points from the missing color that can each be moved onto one of the non-colorful hyperplanes. Then, the two non-colorful hyperplanes become colorful hyperplanes in $P'$ with the same \av so that we are in the setup of Lemma~\ref{lem:not_separated_2} and the claim follows.

Without loss of generality, we assume that the missing color of the two non-colorful hyperplanes is color 1. Let $H_p$ denote one of the non-colorful hyperplanes that passes through some $p_2 \in P_2, \ldots, p_d \in P_d$ and $H_q$ denote another non-colorful hyperplane that passes through some $q_2 \in P_2, \ldots, q_d \in P_d$. 
Recall that they have the same \av and \dv. 
Let $x_p,y_p$ (resp. $x_q,y_q$) be the highest ranked points of $P_1$ under $\prec_1$ on either side of $H_p$ (resp. $H_q$) and let $z_p$ (resp. $z_q$) be the intersection of segment $x_py_p$ (resp. $x_qy_q$) with $H_p$ (resp. $H_q$). By the definition and assumption of \dv, the \dv of $H_p$ and $H_q$ is $||x_p-z_p||=||x_q-z_q||$. 
	
Throughout this proof, we consider $H^+$ to be a closed halfspace while $H^-$ is an open halfspace.
Our definition of $P'$ changes depending on the locations of $x_p,x_q,y_p,y_q$ in the cells $H_p^- \cap H_q^+, H_p^+ \cap H_q^+, H_p^+ \cap H_q^-, H_p^- \cap H_q^-$. When $x_p,x_q \in H_p^+ \cap H_q^+$ and $y_p,y_q \in H_p^- \cap H_q^-$, we have $x_p = y_q$ and $y_p = y_q$. Since $||x_p-z_p|| = ||x_q-z_q||$, we also have $z_p = z_q$, which implies that $H_p \cap H_q$ intersects $\conv(P_1)$. Similar to the proof of Lemma~\ref{lem:not_separated_2}, we find a separating hyperplane $H$ between $P_1$ and $\cup_{ i \in[2..d]}P_i$ if it exists. Then, we project $\cup_{ i \in[2..d]} P_i$ towards $z_p$ onto $H$, in which we have two colorful hyperplanes with the same \av in $\R^{d-1}$, so
we can apply Lemma~\ref{lem:not_separated_2} to the sub-problem in $\R^{d-1}$ and recover a desired colorful set as in Lemma~\ref{lem:not_separated_2}. If $H$ does not exist, by Lemma~\ref{lemma:ws-formats} we can also find a desired colorful set and we are done. 
	
In the following, we consider the case of $x_p,x_q \in H_p^+ \cap H_q^+$ with three sub-cases:
\begin{itemize}
	\item $[y_p \in H_p^- \cap H_q^+ \mbox{ and } y_q \in H_p^- \cap H_q^-]$: there must exist a point $r \in P_1 (\cap H_p^+ \cap H_q^-)$, otherwise $H_p$ and $H_q$ cannot have the same \av. Then, we move $y_p$ towards $x_p$ along segment $x_py_p$ until it hits $H_p$ at $z_p$. Similarly, we move $r$ towards $x_p$ along segment $x_pr$ until it hits $H_q$. We define the resulting point set to be $P'$. We can see that both of the first coordinates of the \av{}s of $H_p$ and $H_q$ (with respect to $P'$) are increased by 1, so they still have the same \av, and now $H_p$ and $H_q$ are colorful hyperplanes in $P'$. By Lemma~\ref{lem:not_separated_2}, we can find a colorful set $\{x_1 \in \conv(P_1'), \ldots, x_d\in \conv(P_d')\}$ that lies on a $(d-2)$-flat. Since $y_p$ and $r$ only moved inside $\conv(P_1)$, $\conv(P'_1) \subseteq \conv(P_1)$ and $P'_i = P_i$ for $i=[2..d]$. 	Since the colorful set is also a certificate for $P$, we are done. 
	
	\item $[y_p \in H_p^- \cap H_q^- \mbox{ and } y_q \in H_p^+ \cap H_q^-]$: the argument is symmetrical to the case above. 
	
	\item $[y_p \in H_p^- \cap H_q^+ \mbox{ and } y_q \in H_p^+ \cap H_q^-]$: we move $y_p$ towards $x_p$ along segment $x_py_p$ until it hits $H_p$ and move $y_q$ towards $x_q$ along segment $x_qy_q$ until it hits $H_q$.
\end{itemize}

Next, we consider $x_p \in H_p^+ \cap H_q^+$ and $x_q \in H_p^- \cap H_q^+$. 
\begin{itemize}
	\item $[y_p \in H_p^- \cap H_q^+ \mbox{ and } y_q \in H_p^- \cap H_q^-]$: we have $x_q = y_p$. We move $x_p$ towards $x_q$ along segment $x_px_q$ until it hits $H_p$ and move $x_q$ towards $y_q$ along segment $x_qy_q$ until it hits $H_q$.
	
	\item $[y_p \in H_p^- \cap H_p^- \mbox{ and } y_q \in H_p^+ \cap H_q^-]$: we move $x_p$ towards $x_q$ along segment $x_px_q$ until it hits $H_p$ and move $x_q$ towards $y_p$ along segment $x_qy_p$ until it hits $H_q$.
	
	\item $[y_p \in H_p^- \cap H_q^+ \mbox{ and } y_q \in H_p^+ \cap H_q^-]$: we have $x_q = y_p$. We move $y_p$ towards $x_p$ along segment $x_py_p$ until it hits $H_p$ and move $y_q$ towards $x_p$ along segment $x_py_q$ until it hits $H_q$.
	
	\item $[y_p \in H_p^- \cap H_q^- \mbox{ and } y_q \in H_p^- \cap H_q^-]$: we have $y_p = y_q$. We move $x_p$ towards $x_q$ along segment $x_px_q$ until it hits $H_p$ and move $x_q$ towards $y_q$ along segment $x_qy_q$ until it hits $H_q$.
\end{itemize}
The case for $x_p \in H_p^+ \cap H_q^-$ and $x_q \in H_p^+ \cap H_q^+$ are symmetrical to those above. 

The last case is $x_p \in H_p^+ \cap H_q^-$ and $x_q \in H_p^- \cap H_q^+$. Basically, the sub-cases are the same as those above except one, which happens for $y_p \in H_p^- \cap H_p^+$ and $y_q \in H_p^+ \cap H_q^-$. In this case, we have $x_p=y_q$ and $x_q = y_p$. We move $x_p$ and $x_q$ towards each other along segment $x_px_q$ and until they hit $H_p$ and $H_q$. Note that the segment $x_p x_q$ may intersect the $(d-2)$-flat $H_p\cap H_q$, but this case is 
also handled by Lemma~\ref{lem:not_separated_2}.
\end{proof}

For the second output ($\V(v) < \V(u) < \V(\S(v))$) of \ref{sol:UV3}, there are two cases to consider. In the first case, if both $v$ and $\S(v)$ correspond to the same \av, then $u$ also has the same \av and its \dv is between that of $v$ and $\S(v)$. Since rotating the hyperplane from $v$ to $\S(v)$ does not pass through $u$, we can find a different hyperplane that is interpolated by $v$ and $\S(v)$ and has the same \dv as $u$. Hence, we apply Lemma~\ref{lem:not_separated_3} again to find a witness of the violation. For the second case that the \av of $\S(v)$ increases in one coordinate by one from that of $v$, since the role of \dv is dominated by the role of {\av} in the potential function, the \dv of $u$ can be arbitrarily large. Therefore, we may not be able to apply the interpolation technique from the first again. We argue that we can transform $P$ to a point set $P'$ satisfying $\conv(P'_i) \subseteq \conv(P_i)$ for all $i\in [d]$, such that the hyperplanes of $v$ and $u$ become colorful. Then, we apply Lemma~\ref{lem:not_separated_2} to show that $P'$ is not well-separated, which also implies that $P$ is not well-separated. The precise statement and proof are given in Lemma~\ref{lem:not_separated_4}.

In Lemma~\ref{lem:UV1} we show how to compute a \ref{sol:GV1} solution from a \ref{sol:UV1} solution.
In \Cref{lem:UV2,lem:UV3} we show how we can compute a \ref{sol:GV1} or \ref{sol:GV2} solution, given a \ref{sol:UV2} or \ref{sol:UV3} solution.
A \ref{sol:GV1} or \ref{sol:GV2} solution can also occur with a \ref{sol:U1} solution that has the incorrect \av, and we show how to do this in Lemma~\ref{lem:U1_violation}.
We show that converting these solutions always takes $\poly(n,d)$ time.
The violations may be detected in either the first sub-problem or the second sub-problem.
Our constructions thus culminate in the promised result:

\begin{theorem}
\label{theorem:main}
	\GHS $\in$ \cUEOPL $\subseteq$ \CLS.
\end{theorem}

\section{Double-wedges and the neighborhood graph}
\label{section:double-wedge}

In this section we formally define the notion of double-wedges and the underlying graph
that is defined using rotations.

\subsection{Double-wedges}
\label{subsection:double-wedge}

An \emph{anchor} is a colorful $(d-1)$-tuple of $P$ which spans a $(d-2)$-flat.
Let $P_i$ denote the missing color in the anchor. 
Then the tuple for the anchor $R=(p_1,\dots,p_{d-1})$ is ordered as
$R=\{p_1\in P_1,\dots,p_{i-1}\in P_{i-1},p_{i}\in P_{i+1},\dots,p_{d-1}\in P_{d}\}$.
An anchor $R$ along with a pair of points $p,q\in P$
such that $p,q\not\in R$ is called a \emph{double-wedge} if all of the following hold:
\begin{itemize}
\item the hyperplane $H_p$ through $R \cup \{p\}$ does not contain $q$.
This implies that the hyperplane $H_q$ through $R \cup \{q\}$ does not contain $p$.

\item if $x,y$ are the highest ordered points of $P_i$ under $\prec_i$ on either sides of $H_p,H_q$,
then $R\cup \{x,y\}$ does not lie on a hyperplane.

\item the hyperplanes $H_p$ and $H_q$ both intersect each of the convex hulls 
$\conv(P_1),\dots,\conv(P_d)$. 
If a hyperplane is colorful, the orientation is defined naturally.
Suppose $H_q$ is non-colorful, then we have $(d-1)$ colors in $R\cup \{q\}$,
so we use $R$ and a point in the convex hull of the missing color to define the orientation as described previously.

\item the intersection of the open halfspaces $H_q^{+}\cap H_p^{-}$ is empty,
that is, it does not contain any point of $P$.
Similarly $H_p^{+}\cap H_q^{-}$ must also be empty.
\end{itemize}
We visualize the anchor as a $(d-2)$-ridge through which $H_p,H_q$ pass through.
A rotation around the anchor changes one hyperplane to the other without passing through any other point of $P$.
Intuitively the double-wedge refers to the space $(H_q^{+}\cap H_p^{-})\cup (H_p^{+}\cap H_q^{-})$ and we use this interpretation several times.
See Figure~\ref{figure:double-wedge-1} for a simple example.

\begin{figure}[ht]
\centering
\includegraphics[width=0.7\textwidth,page=2]{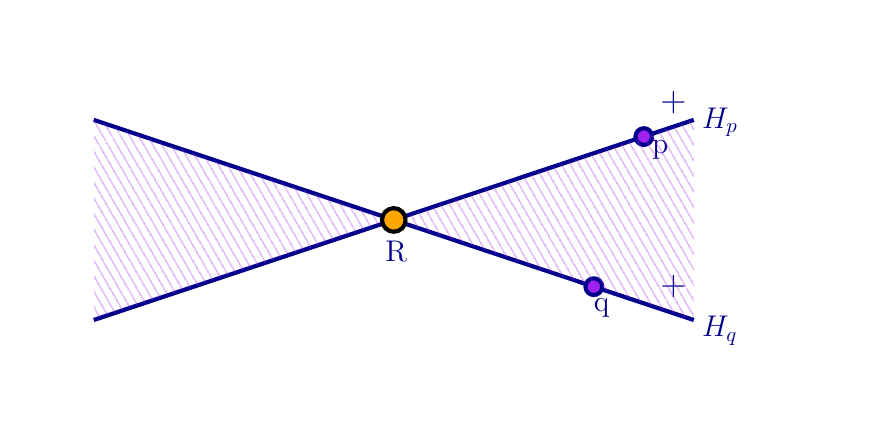}
\caption{A double-wedge $(R,p,q)$ with indicated orientations. The shaded region is empty.}
\label{figure:double-wedge-1}
\end{figure}

For a double-wedge $w:=(R,p,q)$, we define a \emph{representative hyperplane} $H_w$ 
as the hyperplane that is the angular bisector of the double-wedge.
Since a double-wedge is empty, $H_w$ does not contain any point of $P$
apart from $R$.
We define an orientation for $H_w$ based on $R$ and the color missing from $R$.
Let $x,y$ be points from the missing color as defined before.
Without loss of generality, let $x\in H_w^{+}$, and $y\in H_w^{-}$.
We call the first hyperplane among $H_p,H_q$ that intersects the directed segment $xy$ as the \emph{upper hyperplane} of $w$ and the other hyperplane
as the \emph{lower hyperplane} of $w$.
A simple example can be found in Figure~\ref{figure:double-wedge-2}.

\begin{figure}[ht]
\centering
\includegraphics[width=0.7\textwidth,page=5]{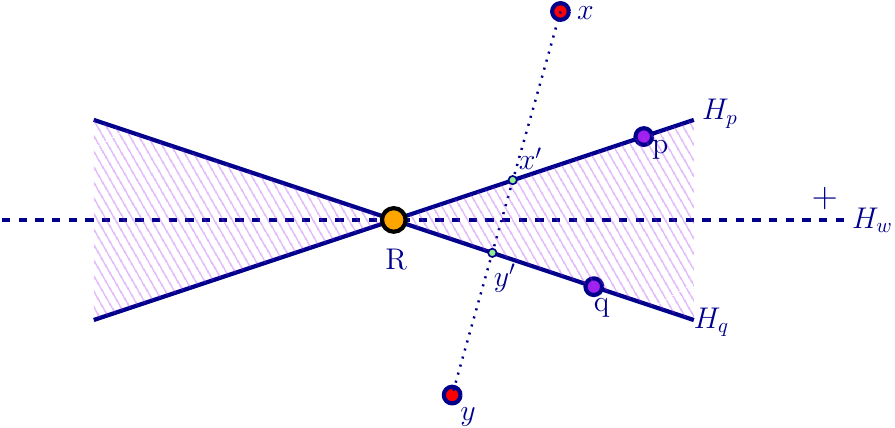}
\caption{The points $x,y$ are from the missing color.}
\label{figure:double-wedge-2}
\end{figure}

The \emph{\av} of any oriented hyperplane $H$ is a $d$-tuple $(a_1,\dots,a_d)$ of integers
where $a_i$ is the number of points of $P_i$ in the closed halfspace $H^{+}$ for $i\in [d]$.
The \av of a double-wedge $w=(R,p,q)$ is defined as the \av of its representative hyperplane.
We say that a double-wedge $w=(R,p,q)$ is 
non-colorful, if both $R\cup \{p\}$ and $R\cup \{q\}$ are non-colorful, and 
colorful, if exactly one of $R\cup \{p\}$ and $R\cup \{q\}$ is colorful, and
very colorful, if both $R\cup \{p\}$ and $R\cup \{q\}$ are colorful.

Under the assumption of weak general position, we additionally have that
if $w$ is non-colorful, then $H_p,H_q$ are non-colorful, and 
if $w$ is colorful, exactly one of $H_p,H_q$ is colorful, and
if $w$ is very colorful, both $H_p,H_q$ are colorful.

\begin{remark} 
The definition of \dv for hyperplanes can be extended to double-wedges by setting the \dv of a double-wedge as that of its representative hyperplane.
Consequently, the results of Lemmas~\ref{lemma:consistent-orientation-0}, \ref{lemma:d-value-strictly} and \ref{lem:inconsistent-orientation-0}
extend to double-wedges with simple modifications. 
\end{remark}

\subsection{Defining a neighborhood graph}
\label{subsection:nbr-graph}

We define a concept of neighborhood between double-wedges, and then we use this to define a graph whose vertices correspond to the double-wedges.
We first describe the graph under the assumptions that the colors are well-separated and $P$ is in weak general position.
Later we show how to handle the cases when these assumptions fail.

We call two double-wedges $(R,p,q),(R',p',q')$ \emph{neighboring} if both share a common hyperplane, that is, 
$\{H_p,H_q\}\cap \{H_{p'},H_{q'}\}\neq\varnothing$, with an exception that we elaborate below.
A double-wedge $w=(R,p,q)$ has different number of neighbors depending on how colorful its hyperplanes are.
The anchor can be written in the form $R=\{x_1,\dots,x_{d-1}\}$. 
\begin{enumerate}
\item \label{NG:case1} Let $w$ be non-colorful. 
Then $p,q$ both share their colors with those of $R$.
Suppose $p$ has the same color as $x_i$.
Then there are at most three neighboring double-wedges that share $H_p$.
One of them use the same anchor $R$, and as an exception we do not count this as a neighboring double-wedge.
For the two remaining neighbors, the anchor is $R'=(x_1,\dots,p,\dots,x_{d-1})$ where $p$ has replaced $x_i$.
The two rotational directions determine the two double-wedges.
Only one of them has the same \av as $w$, since the representative hyperplanes contain $x_i$ on opposite sides.
With a similar argument, there are at most two neighboring double-wedges that share $H_q$ and at most one of them has the same \av as $w$.

\item \label{NG:case2} Let $w$ be colorful, where $H_p$ is colorful and $H_q$ is non-colorful, without loss of generality.
By replacing some $x_i$ by $p$ we get an anchor that is contained in $H_p$ and which may define a double-wedge for each of the two rotational directions.
Since there are $(d-1)$ possible anchors formed by replacement, there are at most $2(d-1)$ double-wedges that share $H_p$. 
Additionally, keeping the anchor $R$ fixed, there is at most one neighboring double-wedge.
So there are at most $2d-1$ neighboring double-wedges of $w$ that share $H_p$.
The case for $H_q$ is similar to case (\ref{NG:case1}).

\item Let $w$ be very colorful. 
Similar to case (\ref{NG:case2}), there are at most $2d-1$ double-wedges sharing $H_p$.
The case for $H_q$ is similar.
\end{enumerate}

\subparagraph*{The neighborhood graph.}
We build a graph $G$ where each vertex corresponds to a double-wedge.
Let $w=(R,p,q)$ be any double-wedge.
For simplicity, we denote the vertex in $G$ corresponding to $w$ also by $w$.
If $H_p$ is colorful, we add an edge in $G$ between $w$ and the vertex of each neighboring double-wedge that shares $H_p$.
If $H_p$ is non-colorful, we add an edge only with the vertex of the double-wedge that shares its \av with that of $w$.
Thus, non-colorful double-wedges have degree two in $G$, while colorful and very colorful double-wedges have degrees at most $2d$
and $4d-2$, respectively.

We transfer each attribute of a double-wedge to its vertex in $G$.
For instance, we call vertices of $G$ as non-colorful, colorful or very colorful corresponding to the color of the double-wedge representing the vertex. 
Similarly, each vertex has an \av that corresponds to the \av of its double-wedge, and so on.
See Figure~\ref{figure:graph} for an elementary example.

\begin{figure}
\centering
\includegraphics[width=0.75\textwidth]{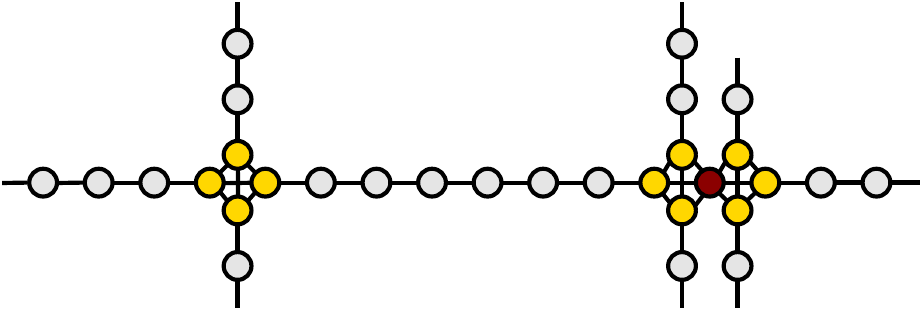}
\caption{A subgraph of $G$: the gray vertices (lightly shaded) are non-colorful, the golden vertices (shaded) are colorful,
and the red vertices (heavily shaded) are very colorful.}
\label{figure:graph}
\end{figure}

Let $v\in G$ be any vertex and let $(\alpha_1,\dots,\alpha_d)$ denote the \av of $v$.
The largest \av for any hyperplane is $(n_1,\dots,n_d)$ that occurs on a unique tangent hyperplane whose half-space contains $P$.
With our definition of the \av of double-wedges using the representative hyperplanes, 
for any double-wedge $w$, the \av of $w$ is smaller in at least one coordinate from the maximum.

\begin{lemma}[Grid-like structure]
\label{lemma:grid-structure}
	Let $H$ be a colorful hyperplane with \av $(\alpha_1,\dots,\alpha_d)$.
	For each $j \in [d]$, if $\alpha_j \leq n_j-1$ (resp. $\alpha_j \geq 2$), then there is a path $w,w_1,w_2,\dots,w_k,w'$ in $G$ such that the double-wedge $w$ is incident to $H$, the double-wedge $w'$ is incident to another colorful hyperplane $H'$, $w,w_1,w_2,\dots,w_k$ share the same \av and the \av for $H'$ differs only in the $j$-th component, where the value is $\alpha_j + 1$ (resp. $\alpha_j-1$).
\end{lemma}

\begin{proof}
For the case that $\alpha_j \leq n_j-1$, we set an anchor $R$ in $H$ that excludes the $j$-th colored point, say $p_j$. Then, we apply {\bf{Procedure}} $NextRotate$ starting from $(R,p_j)$ until we get another colorful hyperplane $H'$. Let $w$ be the double-wedge created by the first rotation. Note that $p_j$ is on the upper hyperplane of $w$.
During this sequence of rotations, we also get a sequence of double-wedges. Before the rotating hyperplane $H''$ hits a point $p_i$ of repeated color $i$, assume that $p_i$ is in the negative half-space of $H''$. Once $p_i$ is on $H''$, we swap $p_i$ with another point $p'_i$ of the same color in the anchor and keep the rotation towards the opposite direction of the orientation so that $p_i$ is in the positive half-space of $H''$ and $p'_i$ from the positive half-space moves to the negative half-space. This is true because the orientation is consistent by Lemma~\ref{lemma:consistent-orientation-0}. If $p_i$ is in the positive half-space of $H''$ before $p_i$ is hit by $H''$, then $p_i$ remains in the positive half-space and $p'_i$ as well (see Figure~\ref{figure:why-double}).
Both cases maintain $\alpha_i$ during the rotation. Thus, all non-colorful double-wedges in this sequence of rotations have the same \av as $w$. In the last rotation, the rotating hyperplane hits the first point $p'_j$ of color $j$. By Lemma~\ref{lem:inconsistent-orientation-0}, well-separation guarantees the consistency of the orientation of the rotating hyperplane so that $p'_j$ moves from the negative half-space to the positive half-space and other points of color $j$ remain in the same sides of the hyperplane. Thus, $\alpha_j$ is increased by one. The same argument also works for the case that $\alpha_j \geq 2$ by using the inverse of {\bf{Procedure}} $NextRotate$. 

Since all the double-wedges created by the first rotation for each $j$ are incident to $H$, they are also connected in $G$ by definition.
\end{proof}

By making use of Lemma~\ref{lemma:grid-structure},
we show that:
\begin{lemma}
\label{lemma:connected-graph}
The neighborhood graph $G$ is connected.
\end{lemma}

\begin{proof}
By Lemma~\ref{lemma:grid-structure}, we know that all (very) colorful double-wedges are connected. For non-colorful double-wedges, we apply {\bf{Procedure}} $NextRotate$ on its lower hyperplane until the rotating hyperplane hits some point of the missing color, which implies that non-colorful double-wedges also connect to some (very) colorful double-wedge.
\end{proof}

The neighborhood graph $G$ imitates a grid in a coarse sense.
There is a "vertex" for every colorful $d$-tuple of $P$, and there are paths connecting these grid vertices.
We showed in Lemma~\ref{lemma:connected-graph} that $G$ is connected.
Therefore, given a target \av $(\alpha_1,\dots,\alpha_d)$, the correct $d$-tuple can be found by starting from some vertex and walking towards the solution.
See Figure~\ref{figure:canonical-path} for an illustration.

\begin{figure}[ht]
\centering
\hspace{0.3cm}
\includegraphics[width=0.85\textwidth]{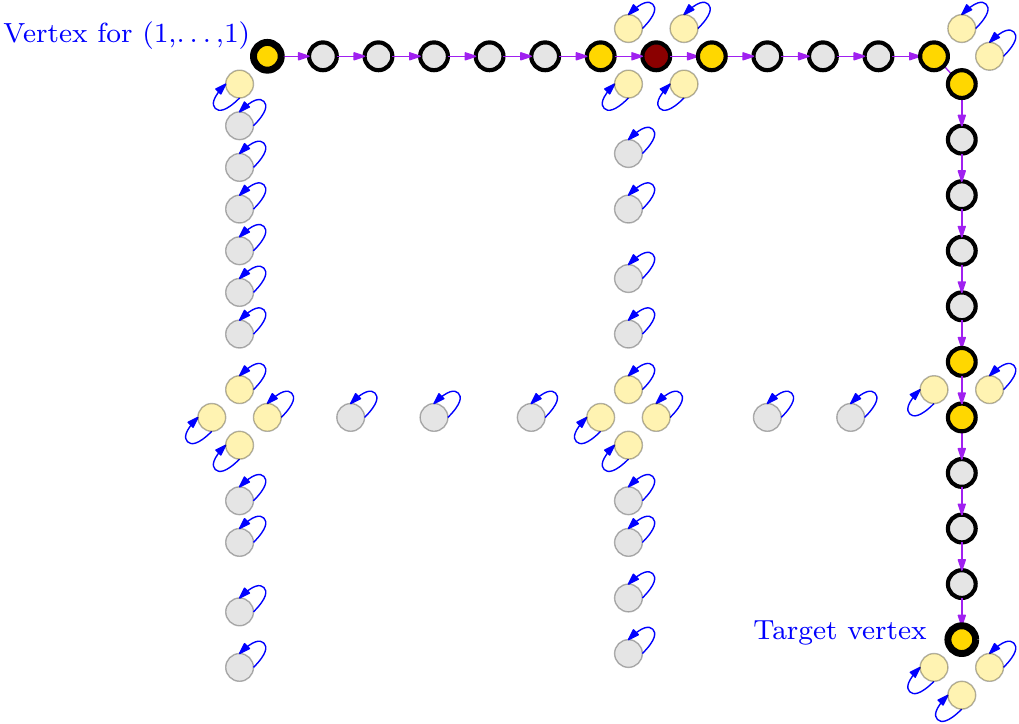}
\caption{The canonical path is marked with arrows. All other vertices in the graph have self-loops.}
\label{figure:canonical-path}
\end{figure}

If $P$ violates well-separation or weak general position, then many nice properties of the neighborhood graph are destroyed.
Double-wedges may fail to have consistent orientations by Lemma~\ref{lem:inconsistent-orientation-0}.
There may be multiple solutions for the same $\alpha$-cut, and no solutions for other cuts.
The former case will manifest as multiple vertices with the same \av, but they may lie in different connected components, 
so Lemma~\ref{lemma:connected-graph} will fail, making the graph disconnected.
For the latter case, there will be no vertex in $G$ that corresponds to the $\alpha$-cut.
The grid-like structure exhibited in Lemma~\ref{lemma:grid-structure} is also not applicable anymore,
meaning that the canonical path may not exist.
See Figure~\ref{figure:violations} for a graph that contains violations.
\begin{figure}[!ht]
\centering
\hspace{0.5cm}
\includegraphics[width=0.85\textwidth]{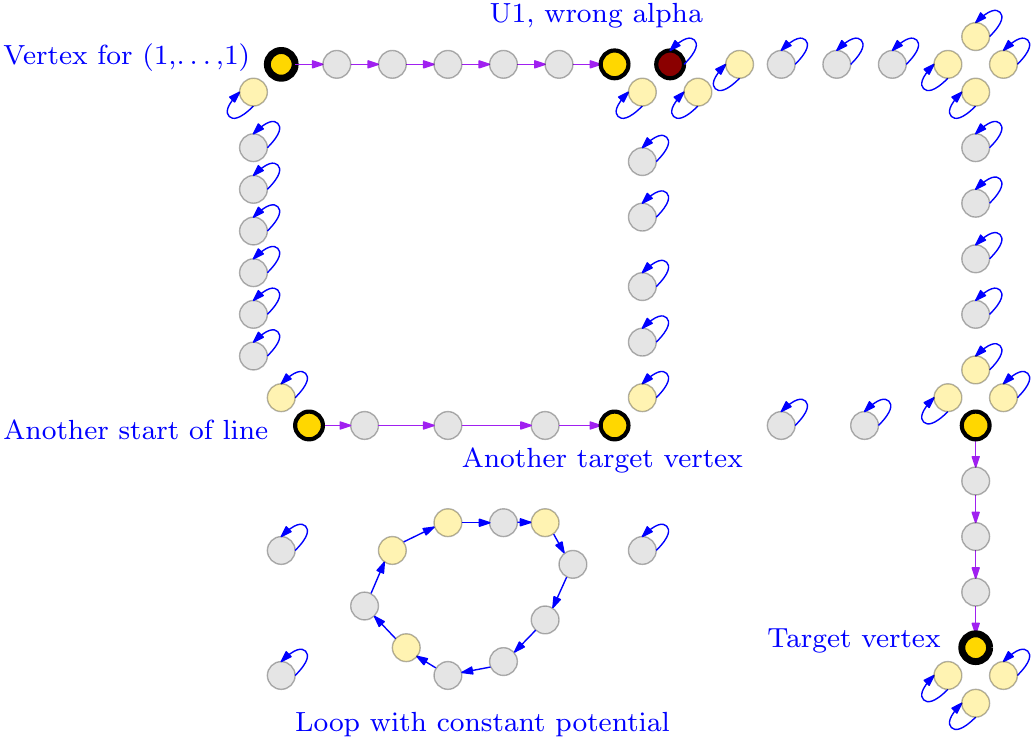}
\caption{A subgraph with multiple violations.}
\label{figure:violations}
\end{figure}

\section{The formal reduction}
\label{section:technical}

In this section, we give a formal reduction from an instance $\cal{I}$ of \GHS to an instance $\cal{I'}$ of \UEOPL in polynomial time. 
An instance $\cal{I}$ of \GHS is defined by $d$ finite sets of points $P = P_1 \cup \ldots \cup P_d$ in $\R^d$ 
and a vector $(\alpha_1,\ldots,\alpha_d)$ of positive integers that satisfy $a_i \leq |P_i|$ for $i=1,\dots,d$. 
Let $m_0$ be the number of bits needed to represent each coordinate of points in $P$ and let $n_0$ denote $\max_{i\in[d]} |P_i|$.
Let $k$ denote the number of coordinates of $\alpha$ that are not equal to one.
Without loss of generality, we assume that $\{ \alpha_1,\dots,\alpha_k \}$ are the non-unit entries in $\alpha$.
Suppose that $(\alpha_1,\ldots,\alpha_d) \neq (1,\ldots,1)$ and 
we are also given a transversal hyperplane $H_0$ passing through some 
$p'_1 \in  P_1,\ldots,p'_d \in P_d$ with $\alpha$-vector $= (1, \ldots, 1)$. 
Later, we show that we do not need to know $H_0$ in advance.
Then, we construct an \UEOPL instance $\cal I'$ on vertex set $\{0, 1\}^\kappa$, where $\kappa=3\cdot \ceil{\log d} + (d+1)\cdot \ceil{\log n_0}$ 
and with procedures $\S$, $\P$ and $\V$ , where $\V : \{0, 1\}^k \rightarrow \{0,\ldots, 2^\gamma-1\}$ and $\gamma = O(\poly(d,n_0,m_0))$.

As shown in Section~\ref{section:ghs-main}, a vertex $v$ in $\cal I'$ corresponds to 
$(R,p,q)$ in $\cal I$, where $R$ is a colorful point set of size $(d-1)$ from $P$, and $p,q \in P$.
We are only interested in the case when $(R,p,q)$ is a double-wedge, as per the
definition in Section~\ref{subsection:double-wedge}.
Otherwise, we create a self loop on $v$ in $\cal{I'}$. 
Furthermore, if there are no violations in $\cal I$, we can define a canonical path from the vertex $v_0$ with \av $= (1,\ldots,1)$ to the unique vertex 
$v_{\alpha}$ with \av $= (\alpha_1, \dots, \alpha_d)$ (shown in Section~\ref{subsection:reduction-overview}), which is the unique path in $\cal I'$. 
For other vertices $v$ not on the path, we also create a self loop on $v$. 
For instance, when $\cal I$ fails weak general position assumption, then $R \cup \{p_1, p_2, \ldots, p_m\} $ lie on the same hyperplane, 
where $p_1 \prec \dots \prec p_m \in P$.
In this case we create a cycle on $v_1=(R,p_1,q), v_2=(R,p_2,q), \ldots, v_m = (R,p_m,q)$ with the same potential value on each $v_i$, 
so that this violation may be reported as the violation~\ref{sol:UV1} in $\cal I'$. 
When $\cal I$ fails the well-separated assumption, the graph we constructed may contain more than one path, which may be reported as the violations~\ref{sol:UV2} or \ref{sol:UV3} in $\cal I'$. In particular, if a hyperplane witnesses both the violations of weak general position and well-separation, then the cycle may become a path with the same potential value, so any violation could be possible.

First we describe how to represent a tuple $(R,p,q)$ in $\kappa$ bits. 
Each vertex is represented as a $(d+4)$-tuple $(t_1, \ldots, t_{d+4}) \in \Z^{d+4}$ such that 
\begin{itemize}
\item $t_1$ contains $\ceil{\log d}$ bits representing the index of the missing color in $R$,

\item $t_2, \ldots, t_d$ each contain $\ceil{\log n_0}$ bits for the indices of the points in $R$ ordered by $\prec$,

\item $(t_{d+1}, t_{d+2})$ contain $\ceil{\log d}$ and $\ceil{\log n_0}$ bits for the index of the color and the index of $p \in P_{t_{d+1}}$ respectively,
\end{itemize}
and we use the same idea to represent $q$ by $(t_{d+3},t_{d+4})$. 
Altogether, we need at most $\kappa=3\cdot \ceil{\log d} + (d+1)\cdot \ceil{\log n_0}$ bits in the encoding.
Let $f_v:\{P\}^{d+1} \rightarrow \{0,1\}^{\kappa}$ denote the function that 
encodes $(R,p,q)$ to $(t_1, \ldots, t_{d+4})$. If $(R,p,q)$ is a double-wedge, then $(R,q,p)$ also represents the same double-wedge. Since we do not want to create two valid vertices in $G$ corresponding to the same double-wedge, we only pick the one with $(t_{d+1},t_{d+2})$ on the upper hyperplane as a valid double-wedge. The following lemma details how we can verify whether a given encoding $(t_1, \ldots, t_{d+4})$ is a double-wedge:

\begin{lemma}
\label{lem:verify_encoding}
	Given a tuple $(t_1, \ldots, t_{d+4})$ and $P$, it can be verified whether it is a double-wedge in $\poly(n,d)$ time.
\end{lemma}
\begin{proof}
We first check whether each $t_i$ is a valid index of the corresponding color or point set. Then, we can set $(R,p,q) = (t_1, \ldots, t_{d+4})$. Next, we check whether $R$ spans a $(d-2)$-flat and the definition of a double-wedge in Section~
\ref{subsection:double-wedge}. Let $x,y$ be the two points from the missing color used to define the orientations of $H_p$ and $H_q$. If the segment $xy$ intersects the affine hull of $R$, then the orientations of $H_p,H_q$ are undefined (if non-colorful) and $R \cup \{x,y\}$ violates well-separation as well as weak general position. In this case, we consider $(R,p,q)$ not to be a double-wedge. If $(R,p,q)$ is supposed to be on the canonical path, then the path is cut into two, which can be detected as a violation in Lemmas~\ref{lem:U1_violation} or \ref{lem:UV2}.
The last step is to confirm that $p$ lies on the upper hyperplane. If any of the above steps fails, then $(t_1, \ldots, t_{d+4})$ is not a double-wedge. It is not hard to see that each step can be done in $\poly(n,d)$ time.
\end{proof}

If there are no violations in $\cal I$, it is straightforward to represent the canonical path using the procedures $\S$ and $\P$. 
The main challenge of the reduction is to handle the violations and to obtain the violation certificate from the output of the \UEOPL instance. 
We first describe two key functions that show how to find the two neighbors of a given vertex $w$ on the canonical path. When $w$ is non-colorful and we want to move forward (resp. backward) along the path, then there is a repeated color point $q$ (resp. $p$) on the lower (resp. upper) hyperplane, so that we can swap that point with the one in $R$ with the corresponding color, and rotate the hyperplane in such a way that the \av is preserved.
When $w$ is (very) colorful, we can choose to swap colored points in $R$ such that the \av is increased/decreased by the definition of the canonical path.

\begin{figure}
\centering
\includegraphics[width=0.7\textwidth,page=3]{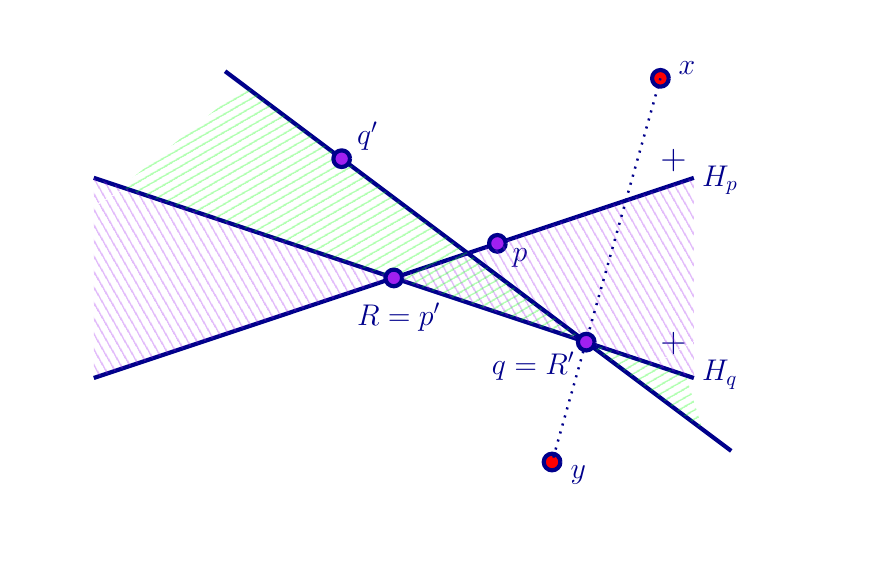}
\caption{From $(R,p,q)$, we rotate $H_q$ anchored at $q=R'$ and get the next $(R',p',q')$, but the orientation of $(R',p',q')$ is not defined because the segment $xy$ passes through $R'$.}
\label{figure:nextN}
\end{figure}

The rotation process is handled by procedures $NextNeighbor$ and $PrevNeighbor$, which we describe next.
Let $w=(R,p,q)$ be the current double-wedge.
Some abnormal cases may happen in the output $w'=(R',p',q')$ of $NextNeighbor$ or $PrevNeighbor$ when the segment $xy$ that defines the orientation of $H_{w'}$ passes through $R'$ (see Figure~\ref{figure:nextN}),  $H_{p'}$ or $H_{q'}$ contains more than $d$ points, or the orientations of $H_{p'},H_{q'},H_{w'}$ are not consistent. For these cases, the path will end or start at $w$. When \UEOPL outputs $w$, we can compute $w'$ and find the certificate of a violation as follows:
for the first or second case, it is easy to see that it violates well-separation and/or weak general position. For the third case, we show that it violates well-separation and that we can obtain a certificate for the violation in Lemma~\ref{lem:inconsistent-orientation-0}.

\vspace{0.15in}
\noindent{\bf{Procedure}} $(R',p',q') = NextNeighbor(R,p,q)$
\begin{enumerate}
\item Let $w=(R,p,q)$ and let $t_1$ be the missing color in $R$.
\item Let $P^+_{t_1}$ be the subset of $P_{t_1}$ that lies in $H^+_w$ and $P^-_{t_1}$ be the subset of $P_{t_1}$ that lies in $H^-_w$. Then, let $x \in P^+_{t_1}$ be the highest ranked point according to the order $\prec_{t_1}$ and $y \in P^-_{t_1}$ be the highest ranked point according to the order $\prec_{t_1}$. As we describe previously, $q$ lies on the lower hyperplane.
\item Since $w$ is supposed to be on the canonical path and is not the end point, the \av of $H_q$ is in the form of $(\alpha_1,\ldots,\alpha_{t_1-1}, b_{t_1},1\ldots,1)$ with $b_{t_1} \leq \alpha_{t_1}$.
\begin{itemize}
\item \label{Next:notColor} If $q$ shares the same color of a point $r$ in $R$, then set $R' := R \cup \{q\} \setminus \{r\}$ and $p' := r$. 
\item If $q$ is in color $t_1$ and $b_{t_1} < \alpha_{t_1}$, then set $R' := R$ and $p' := q$.
\item If $q$ is in color $t_1$ and $b_{t_1} = \alpha_{t_1}$, then let $r$ be the point in $R$ with color $t_1+1$ 
and set $R' := R \cup \{q\}\setminus \{r\}$ and $p'=r$.
\end{itemize}
\item We rotate $H_q$ around the anchor $R'$ in a direction such that the hyperplane is moving away from $x$ along the segment $xy$ until it hits a point $q' \in P$.
\item \textsf{Return} $(R',p',q')$.
\end{enumerate}

\noindent{\bf{Procedure}} $(R',p',q') = PrevNeighbor(R,p,q)$
\begin{enumerate}
	\item Let $w=(R,p,q)$ and let $t_1$ be the missing color in $R$.
	\item Let $P^+_{t_1}$ be the subset of $P_{t_1}$ that lies in $H^+_w$ and $P^-_{t_1}$ be the subset of $P_{t_1}$ that lies in $H^-_w$. Then, let $x \in P^+_{t_1}$ be the highest ranked point according to the order $\prec_{t_1}$ and $y \in P^-_{t_1}$ be the highest ranked point according to the order $\prec_{t_1}$. As we describe previously, $p$ lies on the upper hyperplane.
	\item Since $w$ is supposed to be on the canonical path and is not the starting point, the \av of $H_p$ is in the form of $(\alpha_1,\ldots,\alpha_{t_1-1}, b_{t_1},1\ldots,1)$ with $b_{t_1} \leq \alpha_{t_1}$. When $t_1 = 1$, $b_{1} > 1$.
	\begin{itemize}
	\item If $p$ shares the same color of a point $r$ in $R$, then set $R' := R \cup \{p\} \setminus \{r\}$ and $q' := r$. 
	\item If $p$ is in color $t_1$ and $b_{t_1} > 1$, then set $R' := R$ and $q' := p$.
	\item If $p$ is in color $t_1$ and $b_{t_1} = 1$, then let $r$ be the point in $R$ with color $t_1-1$ and set $R' := R \cup \{p\}\setminus \{r\}$ and $q' = r$.
	\end{itemize}
	\item We rotate $H_p$ around the anchor $R'$ in a direction such that the hyperplane is moving closer to $x$ along the segment $xy$ until it hits a point $p' \in P$.
	\item \textsf{Return} $(R',p',q')$.
\end{enumerate}

\begin{lemma}
The procedures $NextNeighbor(R,p,q)$ and $PrevNeighbor(R,p,q)$ can be completed in $\poly(n,d)$ time.
\end{lemma}

\begin{proof}
The points $x$ and $y$ from the missing color can be found in linear time. We can also check which point will hit the hyperplane first during the rotation by a prune-and-search technique in polynomial time.
\end{proof}

Now we discuss the implementation of $\S$ and $\P$ in $\cal I'$. In {\bf{Procedure}} $\S$, we first point the standard source $0^\kappa$ to the start of the canonical path in Step~\ref{S:w0}. For those vertices $(t_1, \ldots, t_{d+4})$ not on the canonical path, we form self-loops in Steps~\ref{S:invalid}, \ref{S:inconsistent} and \ref{S:notPath}. After these steps, we check whether we reached the target end point of the canonical path in Step~\ref{S:endpt}. To handle the violation of weak general position, when a hyperplane $H$ contains more than $d$ points and at least $d-1$ colors, then it happens that several tuples $(R,p,q)$ represents the same double-wedge. We connect all these tuples to form a cycle in Step~\ref{S:notGeneral_1} when $H$ is the upper hyperplane.
If instead $H$ is the lower hyperplane, we handle this in Step~\ref{S:notGeneral_2}. In the end, we use $NextNeighbor$ to advance to the next vertex along the canonical path in Step~\ref{S:next_1}. As we mentioned previously, if something goes wrong in the next vertex, the path will end in Step~\ref{S:next_2}. The {\bf{Procedure}} $\P$ is basically the same as {\bf{Procedure}} $\S$, so we only mention the differences. In Steps~\ref{P:head} and~\ref{P:2nd_head} of $\P$, we ensure a consistent relation between the standard source $0^\kappa$ and the start of the canonical path. If there exists another (very) colorful double-wedge with \av $(1,\ldots,1)$, then it will not connect to $0^\kappa$ via $\P$, instead it will be the start of another path.

\vspace{0.15in}
\noindent{\bf{Procedure}} $\S(v=(t_1, \ldots, t_{d+4}))$
\renewcommand{\theenumi}{$\bf\S\hspace{0.01cm}{\arabic{enumi}}$}
\begin{enumerate}
	\setlength{\itemindent}{0.4cm}
	\item \label[Step]{S:w0} If $v=0^{\kappa}$, then \textsf{Return} $f_v(\{p'_2,\ldots,p'_d\},p'_1,q')$, where
	$\{p'_1,\ldots, p'_d\}$ is a colorful point set on a hyperplane $H_0$ that has \av $=(1,\ldots,1)$, and
	 $q'$ is the point that creates a double-wedge with the anchor $\{p'_2,\ldots,p'_d\}$ and $p'_1$.
	
	\item \label[Step]{S:invalid} If $(t_1, \ldots, t_{d+4})$ is not the encoding of a double-wedge, then \textsf{Return} $(t_1, \ldots, t_{d+4})$.
	
	\item \label[Step]{S:inconsistent} Let $(R,p,q)$ denote the double-wedge for which $f_v(R,p,q)=(t_1, \ldots, t_{d+4})$.
	If the orientations of $H_w$, $H_p$ and $H_q$ are not consistent, then \textsf{Return} $f_v(R,p,q)$.
	
	\item \label[Step]{S:notPath} Recall that $t_1$ is the missing color in $R$. If the \av of $(R,p,q)$ is not in the form of $(\alpha_1,\ldots,\alpha_{t_1-1},b_{t_1},1,\ldots,1)$ with $b_{t_1} < \alpha_{t_1}$,
	then \textsf{Return} $f_v(R,p,q)$. 
	\item \label[Step]{S:endpt} If the \av of the lower hyperplane $H_q$ is $(\alpha_1,\ldots,\alpha_d)$, then \textsf{Return} $f_v(R,p,q)$.
	
	\item \label[Step]{S:notGeneral_1} If $H_p$ contains some points of $P$ other than $R \cup \{p\}$, then let $p_1, p_2, \ldots, p_m$ be those extra points ordered by $\prec$, and let $p_i$ be the point just after $p$ in the order $\prec$ (if $p$ is after $p_m$, then $p_i=p_1$). We \textsf{Return} $f_v(R,p_i,q)$.

	\item \label[Step]{S:notGeneral_2} Similarly, \textsf{Return} $f_v(R,p,q_i)$ if $H_q$ contains some points other than $R\cup\{q\}$, where $q_i$ is the point just after $q$ according to the order $\prec$ among those extra points.

	\item \label[Step]{S:next_1} Let $w'=(R',p',q')$ be the output of $NextNeighbor(R,p,q)$. If the orientations of $H_{w'}$, $H_{p'}$ and $H_{q'}$ are well-defined and consistent, then \textsf{Return} $f_v(R',p',q')$.
	\item \label[Step]{S:next_2}  Otherwise \textsf{Return} $f_v(R,p,q)$.
\end{enumerate}

\noindent{\bf Procedure} $\P(v=(t_1, \ldots, t_{d+4}))$
\renewcommand{\theenumi}{$\bf\P\hspace{0.01cm}\arabic{enumi}$}
\begin{enumerate}
	\setlength{\itemindent}{0.4cm}
	\item \label[Step]{P:w_0} If $v=0^{\kappa}$, then \textsf{Return} $0^{\kappa}$.
	\item \label[Step]{P:invalid} If $(t_1, \ldots, t_{d+4})$ is not a double-wedge, then \textsf{Return} $(t_1, \ldots, t_{d+4})$.
	\item \label[Step]{P:inconsistent} Let $(R,p,q)$ be a double-wedge such that $f_v(R,p,q)=(t_1, \ldots, t_{d+4})$. If the orientations of $H_w$, $H_p$ and $H_q$ are not consistent, then \textsf{Return} $f_v(R,p,q)$.

	\item \label[Step]{P:head} If $R= \{p'_2,\ldots,p'_d\}$ and $p = p'_1$, then \textsf{Return} $0^\kappa$.
	
	\item \label[Step]{P:2nd_head} If $H_p$ is colorful, $t_1=1$ and the \av of $H_p$ is $(1,\ldots,1)$, then \textsf{Return} $f_v(R,p,q)$. 
	\item \label[Step]{P:notPath} Recall that $t_1$ is the missing color in $R$. If the \av of $(R,p,q)$ is not in the form of $(\alpha_1,\ldots,\alpha_{t_1-1},b_{t_1},1,\ldots,1)$ with $b_{t_1} < \alpha_{t_1}$,
	then \textsf{Return} $f_v(R,p,q)$.

	\item \label[Step]{P:notGeneral_1} If $H_p$ contains some points of $P$ other than $R \cup \{p\}$, then let $p_1, p_2, \ldots, p_m$ be those extra points ordered by $\prec$, and let $p_i$ be the point just before $p$ in $\prec$ (if $p$ is before $p_1$, then $p_i=p_m$), and \textsf{Return} $f_v(R,p_i,q)$.
	
	\item \label[Step]{P:notGeneral_2} Similarly, \textsf{Return} $f_v(R,p,q_i)$ if $H_q$ contains some points other than $R\cup\{q\}$, where $q_i$ is the point just before $q$ in $\prec$ among those extra points.

	\item \label[Step]{P:prev_1} Let $w'=(R',p',q')$ be the output of $PrevNeighbor(R,p,q)$. If the orientations of $H_{w'}$, $H_{p'}$ and $H_{q'}$ are well-defined and consistent, then \textsf{Return} $f_v(R',p',q')$.

	\item \label[Step]{P:prev_2} Otherwise \textsf{Return} $f_v(R,p,q)$.
\end{enumerate}

Given a double-wedge $(R,p,q)$, let $x'$ be any point of $P$ in $H^+_p \cap H^+_q$ and let $y'$ be any point of $P$ in $H^-_p \cap H^-_q$. 
Suppose that the segment $x'y'$ does not pass through the affine hull of $R$ so that the intersections of $x'y'$ with $H_p$ and $H_q$ are two distinct points.
Let $d_{\min}$ be the Euclidean distance between the two intersection points of $x'y'$ with $H_p$ and $H_q$.
Let $d_{\max}$ be the Euclidean distance between $x$ and the intersection point of $x'y'$ with $H_q$.

\begin{lemma}
\label{lem:dist_precision}
Let $m_0$ denote the number of bits needed to represent each coordinate of any point of $P$. 
Then, $d_{\min}$ is at less $1/N^2$ and $d_{\max}$ is at most $M$, where $N = d!2^{dm_0}$ and $M = \sqrt{d}2^{m_0}$.
\end{lemma}

\begin{proof}
Without loss of generality, we assume that the missing color in $R$ is color 1, i.e.,
$R$ is a set of points $\{ p_2 \in P_2, \ldots, p_d \in P_d \}$.
Let $z_p$ (resp. $z_q$)	be the intersection of $x'y'$ and $H_p$ (resp. $H_q$). 
Since $z_p$ lies on $x'y'$ and the affine hull of $H_p$, we can represent $z_p$ as the convex combination of $x'$ and $y'$ and the linear combination of $R\cup\{p\}$ as follows:
\begin{equation*}
z_p = \lambda_1 x' + (1-\lambda_1) y' = p_1 + \sum_{i=2}^d \lambda_i (p_i-p_1), 0 \leq \lambda_1 \leq 1, \lambda_i\in \R~ \forall i\in[2\dots d].
\end{equation*}
Then, we can formulate it as a linear system $A\lambda=b$, where $A \in \Z^{d\times d}$ and $b \in \Z^d$:
\begin{equation*}
\Bigg[y'-x' ~|~ p_2-p_1 ~|~ \ldots ~|~ p_d - p_1\Bigg]\cdot\Bigg[\lambda_1 ~|~ \lambda_2 ~|~ \ldots ~|~ \lambda_d\Bigg]^t = y'-p_1,
\end{equation*}
where $y'-x',p_2-p_1,\dots$ are column vectors.
Since we assume that $z_p$ exists, we have $\det A \neq 0$. According to Cramer's rule, we have $\lambda_i = \frac{\det A_i}{\det A}$, 
where $A_i$ is the matrix obtained by replacing the $i$-th column of $A$ with $b$. Using Leibniz formula for determinants, we can bound the denominator:
\begin{equation*}
	|\det A| = \bigg|\sum_{\sigma \in S_d} \mathrm{sgn}(\sigma)\prod_{i=1}^{d} A_{i,\sigma(i)}\bigg| \leq d!2^{dm_0} = N,
\end{equation*}
where $S_d$ is the set of all permutations of $[d]$ and $\mathrm{sgn}(\cdot)$ is the sign function of permutations.
The same bound is also applied to $|\det A_i|$. Then, $\lambda_1 = \frac{i}{|\det A|}$ for some $0 \leq i \leq |\det A| \leq N$. Similarly, we apply the same argument for $z_q$ with another linear system $A'\lambda' = b'$ such that $z_q = \lambda'_1 x' + (1-\lambda'_1)y'$ and $\lambda'_1 = \frac{i'}{|\det A'|}$. Since $z_p$ and $z_q$ are two distinct points, at least one of their coordinates, namely $j$, have different values. Therefore, $d_{\min} \geq ||z_p - z_q|| \geq \bigg|\frac{i}{|\det A|} - \frac{i'}{|\det A'|}\bigg|\cdot(x'_j-y'_j) \geq \bigg|\frac{|\det A'|i - |\det A| i'}{|\det A|\cdot|\det A'|}\bigg|\cdot(x'_j-y'_j) \geq \frac{1}{N^2}$. The numerator is at least one because all values are integers.

For $d_{\max}$, it is less than the longest possible line segment, so $d_{\max} \leq \sqrt{d}2^{m_0}$.
\end{proof}

Using Lemma~\ref{lem:dist_precision}, we can use $\ceil{\log_2 MN^2}$ bits to represent the \dv of any double-wedge so that no two double-wedges on a path of non-colorful vertices of $G$ have the same \dv. 
We now define the circuit $\V$ in $\cal I'$.
We define a potential function $\delta : (R,p,q) \rightarrow \Z$ that measures the distance of $(R,p,q)$ with \av $(b_1,\ldots,b_d)$ and \dv $D$ to 
the starting vertex that has \av $(1,\ldots,1)$. More precisely, 
\[\delta(R,p,q) = MN^2\cdot \left(\sum_{i=1}^{d} n_0^{d-i}(b_i-1)\right) + \floor{DN^2}.
\]

\vspace{0.15in}

\noindent{\bf Procedure} $\V(v=(t_1, \ldots, t_{d+4}))$
\renewcommand{\theenumi}{$\bf\V\hspace{0.01cm}\arabic{enumi}$}
\begin{enumerate}
	\setlength{\itemindent}{0.4cm}	
	\item If $v = 0^k$, then \textsf{Return} $0$.
	\item If $(t_1, \ldots, t_{d+4})$ is not a double-wedge, then \textsf{Return} $0$. Otherwise, let $(R,p,q)$ be a double-wedge such that $f_v(R,p,q)=(t_1, \ldots, t_{d+4})$.
	\item If $\S(v) = v$ and $\P(v) = v$, then \textsf{Return} $0$.
	\item If $\S(v) \neq v$ or $\P(v) \neq v$, then \textsf{Return} $\delta(R,p,q)$.
\end{enumerate}

The following lemma shows that if there are no violations in \GHS, there are also no violations in the constructed \UEOPL instance. 
This makes the reduction promise-preserving. In particular, we can find the $(\alpha_1,\ldots,\alpha_d)$-cut from the unique solution of the \UEOPL instance.

\begin{lemma}
\label{lem:no_violations}
	If there are no violations in $\cal I$, then the constructed \UEOPL instance $\cal I'$ only contains a type \ref{sol:U1} solution whose lower hyperplane is colorful and has \av $=(\alpha_1, \ldots, \alpha_d)$, which is a type \ref{sol:G1} solution of $\cal I$.
\end{lemma}

\begin{proof}
	First we show that there must exist a type \ref{sol:U1} solution whose lower hyperplane has \av $=(\alpha_1, \ldots, \alpha_d)$.
	If there are no violations in $\cal I$, then by Theorem~\ref{thm:ghs} there exists a unique colorful hyperplane $H$ passing through some $p_1 \in  P_1,\ldots,p_d \in P_d$ with \av $\alpha=(\alpha_1,\ldots,\alpha_d)$. Let $i$ be the largest index of the coordinates in $\alpha$ such that $\alpha_i \neq 1$. Define $R = (p_1, \ldots, p_{i-1},p_{i+1},\ldots,p_d)$. Similar to {\bf Procedure} $PrevNeighbor$, we rotate $H$ around the anchor $R$ in a direction such that $p_i$ is in the open half-space $H^-$ until the hyperplane hits a point $q \in P$. Define $w=(R,q,p_i)$. Since $P$ is well-separated and in weak general position, the orientations of $H_p,H_q,H_w$ are well-defined and consistent. Furthermore, $p_i$ is on the lower hyperplane of $w$. That means, $w$ is a double-wedge. In particular, the \av of $w$ is $(\alpha_1,\ldots,\alpha_i-1,1,\ldots,1)$.
	Similarly, we can also define a double-wedge $w_0$ as shown in Step~\ref{S:w0} of {\bf Procedure} $\S$, which has \av $=(1,\ldots,1)$ and has $H_0$ as the upper hyperplane.
	
	Following the definition of the canonical path in Section~\ref{subsection:reduction-overview}, we can define the canonical path between $w_0$ and $w$, in which every vertex on the path is a double-wedge with \av in the form of $(\alpha_1,\ldots,\alpha_{t_1-1},b_{t_1},1,\ldots,1)$ for some $b_{t_1} < \alpha_{t_1}$, where $t_1 \leq i$ is the missing color of the anchor of the double-wedge, and every two consecutive double-wedges share a hyperplane. This canonical path is realized by Step~\ref{S:next_1} of {\bf Procedure} $\S$ and Step~\ref{P:prev_1} of {\bf Procedure} $\P$. Since the lower hyperplane $H_{p_i}$ has \av $=(\alpha_1,\ldots,\alpha_d)$, {\bf Procedure} $\S(f_v(w))$ will return at Step~\ref{S:endpt} so that $\P(\S(f_v(w))) = \P(f_v(w)) \neq f_v(w)$ and $f_v(w)$ is a type \ref{sol:U1} solution.
	
	Next we will show that there do not exist other solutions. As we mentioned above, we only construct a single path from $0^\kappa$ to $f_v(w_0)$ and then to $f_v(w)$. For other double-wedges not on the path, they will form a self-loop by Step~\ref{S:notPath} of {\bf Procedure} $\S$ and Step~\ref{P:notPath} of {\bf Procedure} $\P$. For non-double-wedges, they will also form a self-loop by Step~\ref{S:invalid} of {\bf Procedure} $\S$ and Step~\ref{P:invalid} of {\bf Procedure} $\P$.
	When $P$ is well-separated and in weak general position, the orientations of all double-wedges are well-defined and consistent as shown in Lemma~\ref{lem:inconsistent-orientation-0}. 
	Therefore, {\bf Procedure} $\S$ will not return at Steps~\ref{S:inconsistent}, \ref{S:notGeneral_1}, \ref{S:notGeneral_2} and \ref{S:next_2}, and {\bf Procedure} $\P$ will also not return at Steps~\ref{P:inconsistent}, \ref{P:notGeneral_1}, \ref{P:notGeneral_2} and \ref{P:prev_2}. Since there is only one colorful hyperplane with \av $=(1,\ldots,1)$, {\bf Procedure} $\P$ will not return at Step~\ref{P:2nd_head}.
\end{proof}

There are different types of violations that may return from a \UEOPL instance. To obtain the certificate of violation~\ref{sol:GV1} in \GHS, we need the following lemmas to convert the violation solutions of \UEOPL to the \ref{sol:GV1} certificate.

\begin{lemma}
	\label{lem:orientation}
	Given $\{ x_1\in \conv(P_1),\dots, x_d\in  \conv(P_d)\}$ that spans a hyperplane $H$.
	For any point $y$ in the open halfspace $H^+$, 
	\[\mathrm{det}
	\begin{vmatrix}
	x_1 & x_2 & \dots & x_d & y \\
	1 & 1 & \dots & 1 & 1
	\end{vmatrix}
	>0 .\]
	Similarly, for any point $y$ in the open halfspace $H^-$, 
	\[\mathrm{det}
	\begin{vmatrix}
	x_1 & x_2 & \dots & x_d & y \\
	1 & 1 & \dots & 1 & 1
	\end{vmatrix}
	<0. \]
\end{lemma}

\begin{proof}
	The normal $\hat{n}$ that defines the orientation of $H$ satisfies $\mathrm{det}
	\begin{vmatrix}
	x_1 & x_2 & \dots & x_d & \hat{n} \\
	1 & 1 & \dots & 1 & 0
	\end{vmatrix}
	>0.$ For any point $y$ in the open halfspace $H^+$, let $\bar{y}$ be the orthogonal projection of $y$ onto $H$, i.e., $\bar{y} =\sum_{i=2}^d \lambda_i (x_i-x_1) + x_1$ and $y = \bar{y} + \gamma \hat{n}$ for some positive number $\gamma$. Then, we have 
\begin{eqnarray*}
& & 
\mathrm{det}
\begin{vmatrix}
 x_1 & x_2 & \dots & x_d & y \\
	1 & 1 & \dots & 1 & 1
\end{vmatrix}\\
& = & 
\mathrm{det}
\begin{vmatrix}
	x_1 & x_2 & \dots & x_d & \bar{y} + \gamma \hat{n} \\
	1 & 1 & \dots & 1 & 1
\end{vmatrix}\\
& = &
\mathrm{det}
\begin{vmatrix}
	x_1 & x_2 & \dots & x_d & \hat{y}-x_1 + \gamma \hat{n} \\
	1 & 1 & \dots & 1 & 0
\end{vmatrix}\\
& = &
\mathrm{det}
\begin{vmatrix}
	x_1 & x_2 & \dots & x_d & \gamma \hat{n} \\
	1 & 1 & \dots & 1 & 0
\end{vmatrix}>0.\\  
\end{eqnarray*}
The proof for points in the open halfspace $H^-$ is similar.
\end{proof}

By Theorem~\ref{thm:ghs}, we already know that if there are multiple $\alpha$-cuts, then $P$ is not well-separated. 
The following two lemmas give the same result, but we provide a constructive proof so that we can find the certificate of the violation.
	
\begin{lemma}
\label{lem:not_separated_1}
Given two colorful hyperplanes $H_p$, $H_q$ that have the same \av such that $H_p \cap H_q$ does not intersect $\conv(P_i)$ for all $i \in [d]$, we can find a colorful set of points $\{x_1 \in \conv(P_1), \ldots, x_d\in \conv(P_d)\}$ that lies on a $(d-2)$-flat in $\poly(n,d,m_0)$ time, which shows that $P$ is not well-separated.	
\end{lemma}

\begin{proof}
The proof is based on the idea of finding a hyperplane that "interpolates" between $H_p$ and $H_q$, for which no consistent orientation can be defined.
This hyperplane gives the certificate of violation.

Let $p_1 \in P_1, \ldots, p_d \in P_d$ denote the colorful points on $H_p$ and $q_1 \in P_1, \ldots, q_d \in P_d$ denote the colorful points on $H_q$.
\begin{claim}
\label{claim:continuous}
	The open segment $p_iq_i$ is contained in either $H_p^+ \cap H_q^+$ or $H_p^- \cap H_q^-$ for each $i\in [d]$.
\end{claim}
\begin{proof}
Since $H_p \cap H_q$ does not intersect $\conv(P_i)$ and both $H_p,H_q$ intersect $\conv(P_i)$, $H_p$ and $H_q$ cut $\conv(P_i)$ into three cells for each $i\in[d]$. Suppose that the middle cell is in $H_p^- \cap H_q^+$. Then, $P_i \cap H_p^+ \subseteq P_i \cap H_q^+$. Since $q_i \notin H_p^+$, then we have $|P_i \cap H_q^+| > |P_i \cap H_p^+|$, which contradicts the assumption that $H_p$ and $H_q$ share their \av{}s. The same argument applies to $H_p^+ \cap H_q^-$. It is straightforward to see that $p_iq_i$ is contained in the middle cell, which is either $H_p^+ \cap H_q^+$ or $H_p^- \cap H_q^-$.
\end{proof}

Let $R$ denote the intersection of $H_p \cap H_q$ and $\theta_0$ denote the angle between $H_p$ and $H_q$ with respect to $H_p^+ \cap H_q^+$.
Let $H(\theta)$ denote the hyperplane that rotates anchored at $R$ from $H_p$ to $H_q$ inside the region $(H_p^+ \cap H_q^+) \cup (H_p^- \cap H_q^-)$ with $\theta$ denoting the angle to $H_p$. By Claim~\ref{claim:continuous}, $H(\theta)$ intersects $p_iq_i$ for all $i\in[d]$. Let $z_i(\theta)$ denote that intersection. We can see that $z_i(\theta)$ changes continuously with $\theta$. Let $r$ be any point of $\R^d$ in the open cell $H_p^+ \cap H_q^-$ that $r$ does not lie on $H(\theta)$ for $0 \leq \theta \leq \theta_0$. 
Notice that $r-z_1(\theta)$ is linearly independent of the vectors $\{z_i(\theta)-z_1(\theta),\forall i\in[2..d]\}$.
Let $Det(\theta) = \mathrm{det}\begin{vmatrix}
z_1(\theta) & z_2(\theta) & \dots & z_d(\theta) & r \\
1 & 1 & \dots & 1 & 1
\end{vmatrix}$, which is a continuous function for $0 \leq \theta \leq \theta_0$. 
Since $r \in H_p^+ \cap H_q^-$, by the definition of the orientation and Lemma~\ref{lem:orientation}, $Det(0) >0$ and $Det(\theta_0) < 0$. From the intermediate value theorem, there exists $0 \leq \theta' \leq \theta_0$ for which $Det(\theta') = 0$. 
Thus, we have
\begin{eqnarray*}
	0 &= & \mathrm{det}\begin{vmatrix}
		z_1(\theta') & z_2(\theta') & \dots & z_d(\theta') & r \\
		1 & 1 & \dots & 1 & 1
	\end{vmatrix}\\
	& = & \mathrm{det}\begin{vmatrix}
		z_1(\theta') & z_2(\theta')-z_1(\theta') & \dots & z_d(\theta')-z_1(\theta') & r-z_1(\theta') \\
		1 & 0 & \dots & 0 & 0
	\end{vmatrix}\\
	& = & (-1)^{d+2}\cdot \mathrm{det}\begin{vmatrix}
		z_2(\theta')-z_1(\theta') & \dots & z_d(\theta')-z_1(\theta') & r-z_1(\theta') \\
	\end{vmatrix}
\end{eqnarray*}
Since $r-z_1(\theta')$ is linearly independent of the vectors $\{z_i(\theta')-z_1(\theta'),\forall i\in[2\dots d]\}$, we see that  $\mathrm{span}(\{z_i(\theta')-z_1(\theta'),\forall i\in[2\dots d]\})$ has dimension at most $d-2$. Hence, the colorful set $\{z_i(\theta')\mid i\in[d]\}$ lies on a $(d-2)$-flat. To compute $\{z_i(\theta') \mid i\in[d]\}$, we use binary search on the interval $[0,\theta_0]$ to find $\theta'$.
\end{proof}

\begin{figure}[!ht]
\centering
\begin{subfigure}[t]{\textwidth}
\centering
\includegraphics[width=0.5\textwidth,page=9]{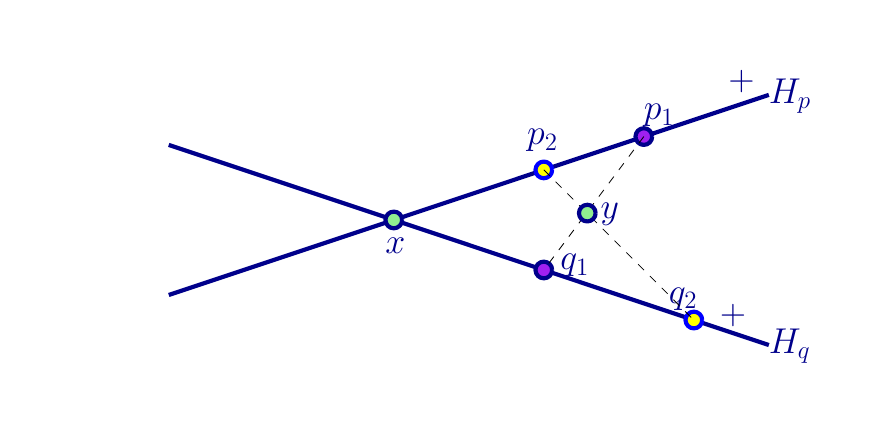}
\caption{The point $y$ is in the convex hulls of colors 1 (purple) and 2 (yellow).}
\label{figure:ns2-1}
\end{subfigure} 
\begin{subfigure}[t]{\textwidth}
\centering
\includegraphics[width=0.5\textwidth,page=7]{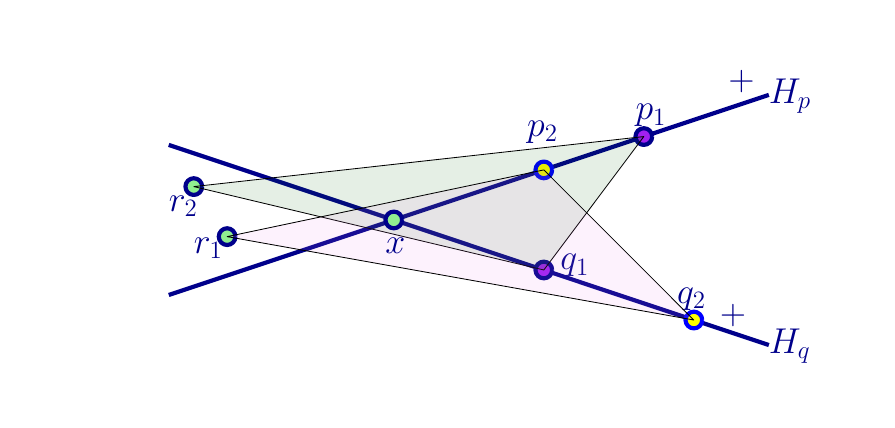}
\caption{The point $x$ lies in both triangles $\triangle p_1q_1r_1$ and $\triangle p_2q_2r_2$.}
\label{figure:ns2-2}
\end{subfigure} 
\begin{subfigure}[t]{\textwidth}
\centering
\includegraphics[width=0.5\textwidth,page=8]{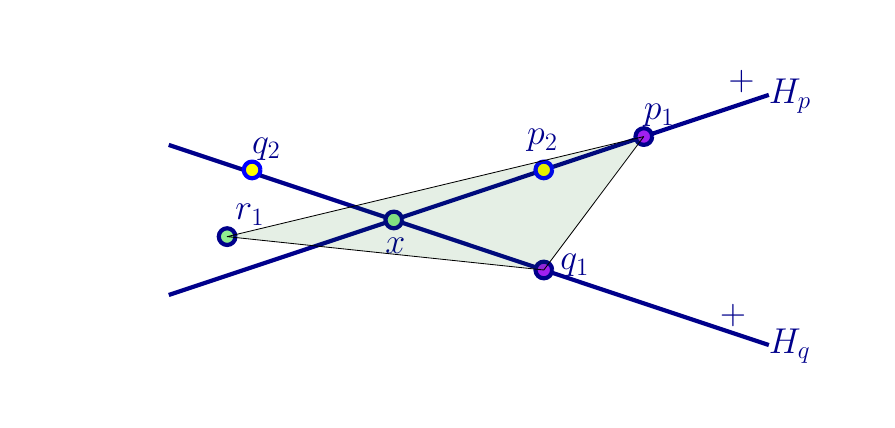}
\caption{The point $p_2$ lies in the triangle $\triangle p_1q_1r_1$.}
\label{figure:ns2-3}
\end{subfigure}
\caption{Some figures for Lemma~\ref{lem:not_separated_2}.}
\label{figure:ns2}
\end{figure}

We now extend Lemma~\ref{lem:not_separated_2} to the case where the two hyperplanes are non-colorful but have the same \av and \dv.
On any path of non-colorful vertices of $G$ that starts and ends at colorful vertices, the \dv{}s of the vertices along the path is bounded by the \dv{}s of the two end points. Hence, we have an interval of the \dv{}s along the path with respect to an \av $(\alpha_1, \ldots, \alpha_{i-1},b_i,1,\ldots,1)$, where $i$ is the missing color for those non-colorful vertices and $b_i < \alpha_i$. If we find another non-colorful vertex with the same missing color and \av but its \dv is outside the interval, we show that it implies a violation of well-separation in the following lemma.

\begin{lemma}
	\label{lem:not_separated_4}
	Given a double-wedge $w=(R,r,q)$ with a colorful lower hyperplane $H_q$ and a non-colorful hyperplane $H$ such that $H_w$ and $H$ have the same missing color and \av, but the \dv of $H$ is larger than that of $H_q$, we can find a colorful set $\{x_1 \in \conv(P_1), \ldots, x_d\in \conv(P_d)\}$ that lies on a $(d-2)$-flat in $\poly(n,d,m_0)$ time.	
\end{lemma}

\begin{proof}
Without loss of generality, we assume that the missing color of $H_w$ and $H$ is color $1$. 
Let $p_2 \in P_2, \ldots, p_d \in P_d$ be the colorful points on $H$ and $q_2 \in P_2, \ldots, q_d \in P_d$ be the colorful points on $H_w$. 
Let $x_p,y_p$ (resp. $x_q,y_q$) be the highest ranked points of $P_1$ under $\prec_1$ on either side of $H$ (resp. $H_w$)
As we saw in the proof of Lemma~\ref{lem:not_separated_3}, the condition that the hyperplanes share the \dv is required only in the first case when $x_p,x_q \in H^+ \cap H_w^+$ and $y_p,y_q \in H^- \cap H_w^-$. We can apply the same analysis from Lemma~\ref{lem:not_separated_3} for $H$ and $H_w$ in other cases.

When $x_p,x_q \in H^+ \cap H_w^+$ and $y_p,y_q \in H^- \cap H_w^-$, we have $x_p = x_q$ and $y_p = y_q$. Since $H_w$ is the representative hyperplane, $H_q$ is the lower hyperplane and the \dv of $H$ is larger than that of $H_q$, 
the directed segment $x_q y_q$ intersects them in the following order: $H_w$, $H_q$ then $H$.
Since $w$ is a double-wedge, the orientations of $H_w$ and $H_q$ are consistent. That means, $q \in H_q \cap H_w^-$. 
There are two possibilities:

\begin{itemize}
\item  $q \in H^+ \cap H_w^-$: using a similar argument as in the proof of Lemma~\ref{lem:not_separated_3}, there must exist another point $r' \in P_1 \cap H^- \cap H_w^+$ because $q \in P_1$ and $H,H_w$ have the same \av. Then, we move $q$ towards $x_p$ along segment $x_pq$ until it hits $H_w$ and move $r'$ towards $x_p$ along segment $x_pr'$ until it hits $H$ (Figure~\ref{figure:ns4-1}). We apply Lemma~\ref{lem:not_separated_2} to the resulting point set and that completes the proof.	

\item $q \in H^{-} \cap H_w^{-}$: in this case both $y_q$ and $q$ lie in $q \in H^{-} \cap H_w^{-}$, but $y_q\neq q$ since that would make the intersection
order along $x_qy_q$ as $H_w,H,H_q$, which is a contradiction.
Along the directed segment $x_qq$, the order of intersection is $H,H_w$ (Figure~\ref{figure:ns4-2}) while on the directed segment $x_qy_q$ it is $H_w,H$.
So we move $y_q$ towards $x_q$ along $x_qy_q$ until it hits $H$, and move $q$ towards $x_q$ along $x_qq$ until it hits $H_w$.
Then we re-use the proof from Lemma~\ref{lem:not_separated_2}.
\end{itemize}
Thus, the claim is proven.
\end{proof}

\begin{figure}[!ht]
\centering
\begin{subfigure}[t]{\textwidth}
	\centering
	\includegraphics[width=0.6\textwidth,page=1]{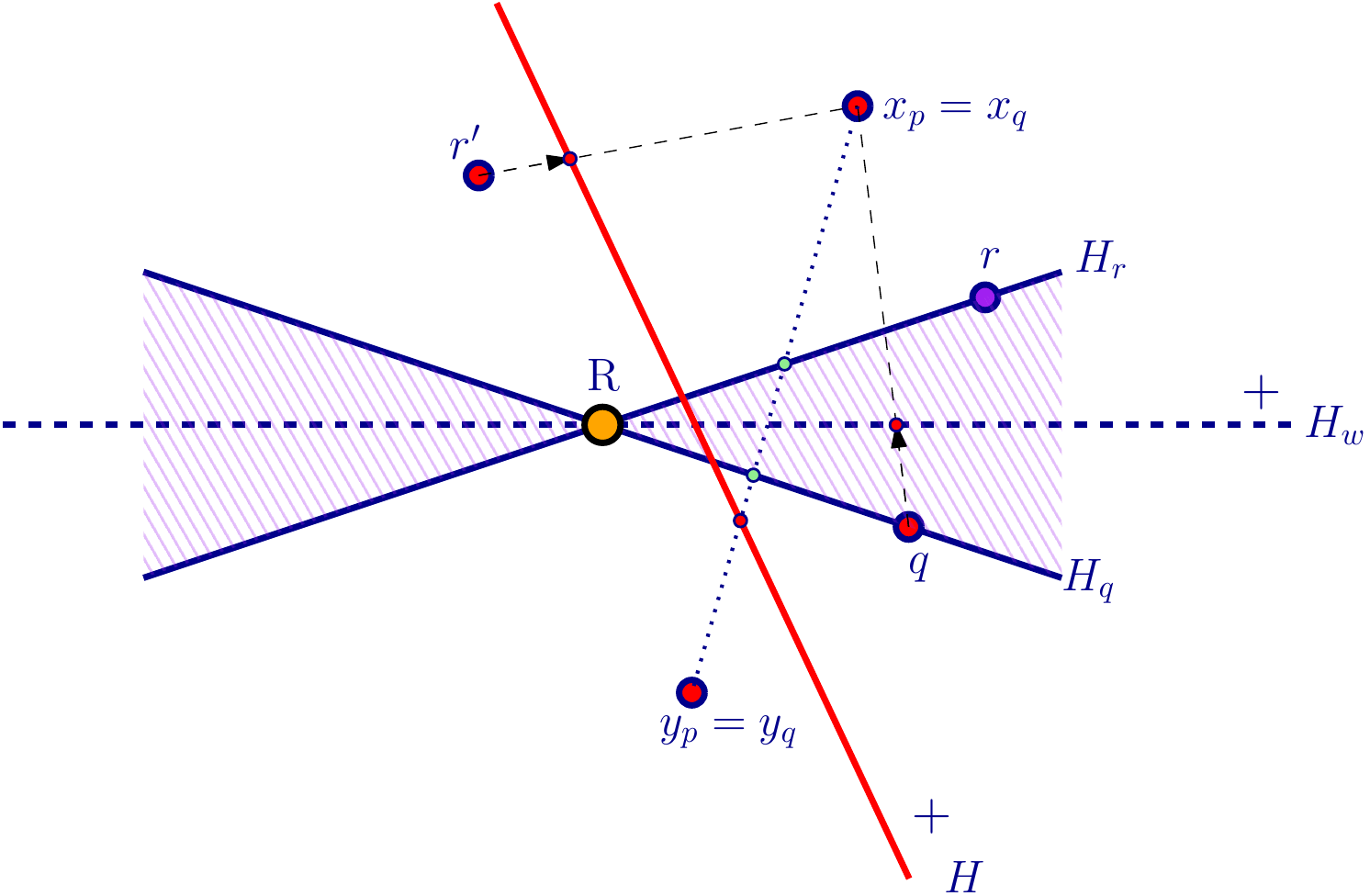}
	\caption{When $q \in H^+ \cap H_w^-$, there exists a point $r' \in P_1 \cap H^- \cap H_w^+$. We move $r'$ and $q$ onto $H$ and $H_w$ respectively.}
	\label{figure:ns4-1}
\end{subfigure}
\begin{subfigure}[t]{\textwidth}
	\centering
	\includegraphics[width=0.6\textwidth,page=2]{violation-double-wedge.pdf}
	\caption{When $q \in H^{-} \cap H_w^{-}$, we move $y_q$ and $q$ onto $H$ and $H_w$ respectively.}
	\label{figure:ns4-2}
\end{subfigure}
\caption{Some figures for Lemma~\ref{lem:not_separated_4}. The small red points represent the new positions after moving those points.}
\label{figure:ns4}
\end{figure}

In fact, the conditions in Lemmas~\ref{lem:not_separated_2}, \ref{lem:not_separated_3} and \ref{lem:not_separated_4} can also serve as certificates of violation of well-separation. Because we want to simplify the statement in \ref{sol:GV2}, from the results of Lemmas~\ref{lem:not_separated_2}, \ref{lem:not_separated_3} and \ref{lem:not_separated_4} we find a colorful set that lies on a $(d-2)$-flat, and then we apply Lemma~\ref{lemma:ws-formats} to compute the index sets $I,J$ as a type \ref{sol:GV2} solution. The reason for not choosing a colorful set as a certificate is that the representation of $I,J$ needs fewer bits than that for a colorful set. Moreover, the bit representation of the violating colorful set may not guarantee that they lie on a $(d-2)$-flat because of the rounding error. Since most of the computations are done by solving linear systems, we believe that a violating colorful set can be represented by $\poly(n,d,m_0)$ bits. If this fails, we could find a set of feasible solutions that lie very close to a $(d-2)$-flat. Then, we can project the set onto the $(d-2)$-flat and one could argue that Lemma~\ref{lemma:ws-formats} would give a correct partition $(I,J)$.

In the following lemmas, we show how to reduce any violation solution of $\cal I'$ to a violation solution of $\cal I$ to complete the reduction.

\begin{lemma}
	\label{lem:U1_violation}
	Let $v$ be a type \ref{sol:U1} solution of the constructed \UEOPL instance $\cal I'$ for which the lower hyperplane has $\alpha$-vector different from the target vector $(\alpha_1,\ldots,\alpha_d)$. Then, we can compute a type \ref{sol:GV1} or \ref{sol:GV2} solution of \GHS instance $\cal I$ in $\poly(n,d,m_0)$ time.
\end{lemma}

\begin{proof}

	By the definition of \ref{sol:U1}, we have $\P(\S(v)) \neq v$, so $v$ is not on a self-loop or a cycle. Also, $v$ cannot be $0^\kappa$. 
	Therefore, $\S(v)$ cannot be returned at \Cref{S:w0,S:invalid,S:inconsistent,S:notPath}.
	Since the lower hyperplane has $\alpha$-vector $\neq (\alpha_1,\ldots,\alpha_d)$, $\S(v)$ also cannot be returned at Step~\ref{S:endpt}.
		
	There is a special case that Step~\ref{S:notGeneral_1} (\ref{S:notGeneral_2}) does not form a cycle, which happens when the hyperplane that contains more than $d$ points and $d$ colors also witnesses a violation of well-separation. The orientation could flip when moving along the path that contains the vertices corresponding to the same double-wedge, so the output $\S(v)$ from Step~\ref{S:notGeneral_1} (\ref{S:notGeneral_2}) is not a valid double-wedge. We can verify this case by checking whether the upper or lower hyperplane of $\S(v)$ contains more than $d$ points, which is a type \ref{sol:GV1} solution.
	
	Now we consider Steps~\ref{S:next_1} and \ref{S:next_2}. There are two possibilities to get $\P(\S(v)) \neq v$: 
\begin{itemize}
\item Case~1: $\S(v) = u \neq v$ and $\P(u) \neq v$, and

\item Case~2: $\S(v) = v$ and $\P(v) = u \neq v$.
\end{itemize}
For Case~1, $u$ is returned at Step~\ref{S:next_1}, which implies that the orientations of $v$ and $u$ are consistent.
Next we consider where $\P(u)$ is returned. From the definition of $NextNeighbor$ and the if-condition of Step~\ref{S:next_1}, $u$ is a double-wedge with consistent orientations and on the canonical path. Hence, $\P(u)$ cannot be returned at \Cref{P:w_0,P:invalid,P:inconsistent,P:head,P:2nd_head,P:notPath}. 
If $\P(u)$ is returned at Step~\ref{P:prev_1} or \ref{P:prev_2}, we claim that it is not possible.
According to {\bf{Procedure}} $NextNeighbor$ and {\bf{Procedure}} $PrevNeighbor$, since the rotational direction is determined by the orientation of the representative hyperplane, $PrevNeighbor(u)$ should return $v$. Then, we have $\P(u)=v$, which contradicts with the assumption of Case~1.
Therefore, the only possibilities of returning are Steps~\ref{P:notGeneral_1} and \ref{P:notGeneral_2}. Then, we check whether the upper or lower hyperplane of $u=\S(v)$ contains more than $d$ points, which is a type \ref{sol:GV1} solution.

For Case~2, let $w'=(R',p',q')$ be the output of $NextNeighbor(v)$. Since $\S(v) = v$, which is returned at Step~\ref{S:next_2}, the orientations of $H_{p'}$, $H_{q'}$ and $H_{w'}$ are either not well-defined or inconsistent. For the case that the orientations are not well-defined, 
the segment $xy$ from the missing color of $R'$ intersects the flat of $R'$, so that the intersection point and $R'$ are a type \ref{sol:GV1} solution as well as the index set of the missing color in $R'$ and the index set of other colors are a type \ref{sol:GV2} solution. For the other case of inconsistent orientations, by Lemma~\ref{lem:inconsistent-orientation-0}, we can find a type \ref{sol:GV2} solution.
\end{proof}

The next lemma is to handle type \ref{sol:UV1} solutions that capture a vertex at which the potential value is not increasing. From the way we construct the graph, it only happens when the weak general position fails.

\begin{lemma}
\label{lem:UV1}
	Let $v$ be a type \ref{sol:UV1} solution of the constructed \UEOPL instance $\cal I'$. Then, a type \ref{sol:GV1} solution of $\cal I$ can be computed in $\mathrm{poly}(n,d)$ time.
\end{lemma}

\begin{proof}
	By the definition of \ref{sol:UV1}, we have $\S(v) \neq v$, $\P(\S(v)) = v$ and $\V(\S(v))-\V(v)\le 0$ so $v$ is neither on a self-loop nor an end point of a path. Hence, $\S(v)$ cannot be returned at \cref{S:invalid,S:inconsistent,S:notPath,S:endpt,S:next_2}. The remaining possibilities are \cref{S:w0,S:notGeneral_1,S:notGeneral_2,S:next_1}. Since $\S(0^\kappa)$ has \av $=(1,\ldots,1)$ and \dv $>0$, $\V(\S(0^\kappa)) > \V(0^\kappa)$, so \cref{S:w0} is also not possible. 
	From the way we rotate the hyperplane in {\bf{Procedure}} $NextNeighbor$ and the consistent orientations of $v$ and $\S(v)$ that are confirmed by a return from \cref{S:next_1}, we can see that $\V(\S(v)) > \V(v)$, that is, Step~\ref{S:next_1} is not possible.
	Overall, $\V(\S(v))-\V(v) \leq 0$ can only happen when $\S(v)$ returns at Step~\ref{S:notGeneral_1} or \ref{S:notGeneral_2}, which implies that $v$ and $\S(v)$ represent the same double-wedge geometrically so that they have the same potential value. 
	Let $f_v(R,p,q)=v$. By checking $H_p$ and $H_q$, at least one of them would contain another point in $P$, which is a type \ref{sol:GV1} solution.
\end{proof}

A solution of type \ref{sol:UV2} means that there is another starting point of some other path. The proof of the following lemma is basically the same as Lemma~\ref{lem:U1_violation}. 

\begin{lemma}
	\label{lem:UV2}
	Let $v$ be a type \ref{sol:UV2} solution of the constructed \UEOPL instance $\cal I'$. Then, a type \ref{sol:GV1} or \ref{sol:GV2} solution of $\cal I$ can be computed in $\mathrm{poly}(n,d,m_0)$ time.
\end{lemma}

\begin{proof}
	By the definition of \ref{sol:UV2}, we have $\S(\P(v)) \neq v \neq 0^n$, which implies that $v$ is the starting point of another path. 
	Following the same reasons as in Lemma~\ref{lem:U1_violation}, we can show that $\P(v)$ cannot be returned at \cref{P:w_0,P:invalid,P:inconsistent,P:head,P:notPath}.
	The same special case can also happen at Step~\ref{P:notGeneral_1} or \ref{P:notGeneral_2}. Hence, we check whether the upper or lower hyperplane of $\P(v)$ contains more than $d$ points, which is a type \ref{sol:GV1} solution.
	
	Let $f_v(R,p,q)=v$. When $\P(v)$ is returned at Step~\ref{P:2nd_head}, $H_p$ is colorful with \av $(1,\ldots,1)$ and the missing color of $R$ is color $1$. Since $\P(v)$ did not return at Step~\ref{P:head}, the standard starting hyperplane $H_0$ and $H_p$ are two different colorful hyperplanes with the same \av. Applying Lemma~\ref{lem:not_separated_2} and then Lemma~\ref{lemma:ws-formats}, we can find the index sets $I,J$ as a type \ref{sol:GV2} solution.
	
	For the case of Steps~\ref{P:prev_1} and \ref{P:prev_2}, it is similar to the proof of Lemma~\ref{lem:U1_violation}.
	There are two possible cases in which $\S(\P(v)) \neq v \neq 0^n$: 
\begin{itemize}
	\item Case~1: $\P(v) = u \neq v$ and $\S(u) \neq v$, and 
	
	\item Case~2: $\P(v) = v$ and $\S(v) = u \neq v$.
\end{itemize}	
	For Case~1, we apply the same argument in the proof of Lemma~\ref{lem:U1_violation} to show that the case is impossible.
	For Case~2, let $w'=(R',p',q')$ be the output of $PrevNeighbor(v)$. The orientations of $w'$ are not well-defined or inconsistent, then we can find a type \ref{sol:GV2} solution.
\end{proof}

In \ref{sol:UV3}, either we have two representative hyperplanes with the same potential value, or we find a hyperplane whose potential value is between the potential values of two consecutive representative hyperplanes in $G$. We use Lemmas~\ref{lem:not_separated_3} and \ref{lem:not_separated_4} to find a violation solution of \GHS.

\begin{lemma}
\label{lem:UV3}
	Let $v$ and $u$ be a type \ref{sol:UV3} solution of the constructed \UEOPL instance $\cal I'$. Then, a type \ref{sol:GV1} or \ref{sol:GV2} solution of $\cal I$ can be computed in $\mathrm{poly}(n,d,m_0)$ time.
\end{lemma}

\begin{proof}
By the definition of \ref{sol:UV3}, we have $v\neq u$, $\S(v) \neq v$, $\S(u) \neq u$, and either $\V(v) = \V(u)$ or $\V(v) < \V(u) < \V(\S(v))$, so $v,u$ are not on self-loops, that is, $\S(v)$ and $\S(u)$ are not returned at \Cref{S:invalid,S:inconsistent,S:notPath}. Hence, $v$ may be $0^\kappa$ or the \av{}s of $v,u$ are in the form of $(\alpha_1,\ldots, \alpha_{t_1-1},b_{t_1},1,\ldots,1)$, where $t_1$ is the missing color of the anchor $R$ and $b_{t_1} < \alpha_{t_1}$.
For the first case $\V(v) = \V(u)$, $v$ cannot be $0^\kappa$.
$v$ and $u$ may represent the same double-wedge because of the violation of weak general position. We check whether the upper or lower hyperplanes of $v,u$ contain more than $d$ points in $P$. If yes, we can return $d+1$ points on that hyperplane as a type~\ref{sol:GV1} solution. Otherwise, from the way we define $\V$ we know that $H_v$ and $H_u$ have the same missing color, \av and \dv. By Lemma~\ref{lem:not_separated_3} and then Lemma~\ref{lemma:ws-formats}, we can find a type~\ref{sol:GV2} solution.

For the second case $\V(v) < \V(u) < \V(\S(v))$, $\S(v)$ can only be returned by Step~\ref{S:w0} or \ref{S:next_1} to increase the potential value $\V(v) < \V(\S(v))$. 
The idea is to find a hyperplane $H$ inside the double-wedge $v$ or $\S(v) $ such that the \dv of $H$ is the same as that of $H_u$. 
When $v=0^\kappa$, both $\S(v)$ and $u$ have \av $(1,\ldots,1)$ and the \dv of $u$ is between $0$ and the \dv of $H_{S(v)}$. Let $f_v(R',p',q') = \S(v)$. During the rotation of a hyperplane from $H_{p'}$ to $H_{S(v)}$ anchored at $R'$, we can find an intermediate hyperplane $H$ that has the same missing color, \av and \dv as $H_u$. 
If $v$ and $\S(v)$ share a non-colorful hyperplane,
then $H_v$ and $H_{\S(v)}$ have the same missing color and \av. By the definition of $\V$, $H_u$ also has the same missing color and \av. 
Therefore, the \dv of $H_u$ is between the \dv{}s of $H_v$ and $H_{\S(v)}$. 
We apply the same idea as above to rotate a hyperplane within the double-wedges $v$ and $\S(v)$ to find a hyperplane $H$ that has the same missing color, \av and \dv as $H_u$. 
For the other case when both $v$ and $\S(v)$ share a colorful hyperplane, the $t_1$-th coordinate of the \av of $\S(v)$ is increased by 1 comparing to that of $v$, where $t_1$ is the missing color of the anchor $R$ of $v$. If $u$ has the same \av as $\S(v)$, then the \dv of $H_u$ is between $0$ and the \dv of $H_{\S(v)}$. We can find a hyperplane $H$ inside the double-wedge $\S(v)$ that has the same missing color, \av and \dv as $H_u$. If $u$ has the same \av as $v$ and the \dv of $u$ is at most the \dv of the lower hyperplane of $v$, we can still find $H$ in the same way as above.
In all these cases, we can apply Lemma~\ref{lem:not_separated_3} on $H$ and $H_u$ and then apply Lemma~\ref{lemma:ws-formats} to find a type~\ref{sol:GV2} solution.
Otherwise, we have that the \dv of $u$ is larger than the \dv of the colorful lower hyperplane of $v$. Now $v$ and $H_u$ satisfy the condition of Lemma~\ref{lem:not_separated_4} and we again apply Lemma~\ref{lemma:ws-formats} to find a type~\ref{sol:GV2} solution.
\end{proof}

Now we discuss how to find $H_0$. We pick any transversal hyperplane $H''$ passing through some colorful set $p''_1 \in  P_1,\ldots,p''_d \in P_d$.
Let  $b=(b_1,\ldots,b_d)$ denote the \av for $H''$. Without loss of generality, we rearrange the order of the colors such that all $b_i \neq 1$ occupy the first few coordinates of the \av. Let $k$ denote the number of coordinates of $b$ that not equal to one. Let $R = (p''_1,\ldots,p''_{k-1},p''_{k+1},\ldots,p''_d)$. We search for a point $q\in P$ such that $(R,p''_{k},q)$ is a double-wedge with \av $b=(b_1,\ldots,b_d)$. Basically, we reverse the roles of $H_0$ and $H''$ as in the above reduction so that the standard vertex $0^\kappa$ connects to $v_b=f_v(R,p''_{k},q)$ and we look for the end vertex $v_0$ corresponding to $H_0$.
If there are no violations in $\cal I$, there exists a canonical path between $v_b$ and $v_0$. Similarly, we reverse the roles of $\S$ and $\P$, and $\V$ is now the function that measures the distance of the current vertex $(R,p,q)$ with \av and \dv to the starting vertex $v_b$ that has \av $(b_1,\ldots,b_d)$.
If $H_0$ cannot be found from the reduction, we can apply the same argument in \Cref{lem:U1_violation,lem:UV1,lem:UV2,lem:UV3} to output a type \ref{sol:GV1} or \ref{sol:GV2} solution of $\cal I$.

\section{Conclusion and future work}
\label{section:conclusion}

We gave a complexity-theoretic upper bound for \GHS.
No hardness results are known for this search problem, and the next question is determining if this is hard for \cUEOPL.
One challenge is that \UEOPL is formulated as Boolean circuits, whereas \GHS is purely geometric.
Emulating circuits using purely geometric arguments is highly non-trivial.
Filos-Ratsikas and Goldberg showed a reduction of this form in~\cite{fg-hamsandwich}.
They reduced the \PPA-complete Tucker circuit to \HS, going via the \emph{Consensus-Halving}~\cite{ss-consensus}, 
and the \emph{Necklace-splitting problems}~\cite{aw-necklace}.
It could be a worthwhile exercise to investigate if their techniques can provide insights for hardness of \GHS.

Some related problems are determining the complexity of answering whether a point set is well-separated,
whether it is in weak general position, or whether a given $\alpha$-cut exists for the point set.
A given $\alpha$-cut may exist even when both assumptions are violated.
On a related note, deciding whether the Linear Complementarity problem has a solution is \NP-complete~\cite{chung-lcp}.
The solution is unique if the problem involves a $P$-matrix, but checking this condition is \coNP-complete~\cite{coxson-matrix}.
However, using witnesses to verify whether a matrix is P-matrix or not, a total search version
is shown to be in \cUEOPL.
Our result for \GHS would go in a similar vein, if the complexities of the above problems were better determined.

Another line to work could be to determine the computational complexities of other extensions of the Ham-Sandwich theorem.
For other geometric problems that are total and admit unique solutions, it could be worthwhile to explore their place in the class \cUEOPL.
Faster algorithms for computing the $\alpha$-cut can also be explored.

\newcommand{\SortNoop}[1]{}

\end{document}